\documentclass[preprint,1p,12pt]{elsarticle}

\usepackage{xcolor,tikz}
\usetikzlibrary{arrows}
\usetikzlibrary{datavisualization}
\usetikzlibrary{shapes}
\usepackage{hyperref,soul}
\usepackage{amssymb,amsmath}
\usepackage[utf8]{inputenc}
\usepackage{graphicx,color}
\usepackage{cancel}
\usepackage{nicefrac}
\usepackage{enumitem}

\newcommand{\fcon}[2]{\langle #1,#2\rangle}
\newcommand{\ot}{\leftarrow}
\newcommand{\uimp}{\nwarrow}
\newcommand{\dimp}{\swarrow}
\newcommand{\adjoint}{\mathop{\&}\nolimits}
\newcommand{\Pp}{\text{P}}
\newcommand{\G}{\text{G}}

\newcommand{\B}{\text{B}}
\newcommand{\cM}{\mathcal{M}}
\newcommand{\cC}{\mathcal{C}}
\newcommand{\cF}{\mathcal{F}}
\newcommand{\uN}{{\uparrow_{N}}}
\newcommand{\dN}{{\downarrow^{N}}}

\newcommand{\up}{\uparrow}
\newcommand{\down}{\downarrow}

\newcommand{\oo}{\text{o}}
\newcommand{\po}{\text{p}}

\newcommand*{\bigchi}{\mbox{\scalebox{1.3}{$\chi$}}}

\usepackage{pgfplots}
\pgfplotsset{compat=newest}

\makeatletter
 \pgfplotsset{
    xtick parsed/.code={
        \c@pgf@counta 0\relax
        \foreach \x in {#1} {
            \pgfmathparse{\x}
            \ifnum\c@pgf@counta=0
                \xdef\pgfplots@xtick{\pgfmathresult}
            \else
                \xdef\pgfplots@xtick{\pgfplots@xtick,\pgfmathresult}
            \fi
            \global\advance\c@pgf@counta 1\relax
        }
    }
 } 
 
 \pgfplotsset{
    ytick parsed/.code={
        \c@pgf@counta 0\relax
        \foreach \x in {#1} {
            \pgfmathparse{\x}
            \ifnum\c@pgf@counta=0
                \xdef\pgfplots@ytick{\pgfmathresult}
            \else
                \xdef\pgfplots@ytick{\pgfplots@xtick,\pgfmathresult}
            \fi
            \global\advance\c@pgf@counta 1\relax
        }
    }
 }

 \makeatother

\journal{Computational and Applied Mathematics}

\newtheorem{theorem}{Theorem}
\newtheorem{corollary}[theorem]{Corollary}
\newtheorem{lemma}[theorem]{Lemma}
\newtheorem{proposition}[theorem]{Proposition}
\newdefinition{definition}[theorem]{Definition }
\newdefinition{remark}[theorem]{Remark }
\newdefinition{example}[theorem]{Example }
\newproof{proof}{Proof}

\makeatletter
\def\@author#1{\g@addto@macro\elsauthors{\normalsize%
    \def\baselinestretch{1}%
    \upshape\authorsep#1\unskip\textsuperscript{%
      \ifx\@fnmark\@empty\else\unskip\sep\@fnmark\let\sep=,\fi
      \ifx\@corref\@empty\else\unskip\sep\@corref\let\sep=,\fi
      }%
    \def\authorsep{\unskip,\space}%
    \global\let\@fnmark\@empty
    \global\let\@corref\@empty  
    \global\let\sep\@empty}%
    \@eadauthor={#1}
}
\makeatother

\begin{document}

\begin{frontmatter}

\title{
Decomposition of  contexts into independent subcontexts based on thresholds\tnoteref{t1}  
}
\tnotetext[t1]{Partially supported by the project  PID2022-137620NB-I00 funded by MICIU/AEI/10.13039/501100011033 and FEDER, UE, by the grant TED2021-129748B-I00 funded by MCIN/AEI/10.13039/501100011033 and European Union NextGenerationEU/PRTR, and by the project PR2023-009 funded by the University of C\'adiz.}

\author
{Roberto G. Arag\'on\corref{t2}}
\cortext[t2]{Corresponding author}
\ead{roberto.aragon@uca.es}
\author
{Jes\'us Medina}
\ead{jesus.medina@uca.es}
\author
{Elo\'isa Ram\'irez-Poussa}
\ead{eloisa.ramirez@uca.es}

\address
{Department of Mathematics,
 University of  C\'adiz. Spain\\
}

\begin{abstract}
The process of decomposing databases into smaller datasets, with the objective of extrapolating the information obtained in the smaller ones to the original database, represents a relevant and complex challenge in real applications. 
It is particularly relevant in the context of fuzzy formal concept analysis, where the complexities of knowledge extraction from datasets characterized by incomplete and imperfect data are considerable. 

This paper will analyze a mechanism and different properties for detecting independent subcontexts from a given context, using modal operators within the multi-adjoint concept lattice framework.
\end{abstract}


\end{frontmatter}
\section{Introduction}

Formal concept analysis~\cite{GanterW} is
a mathematical tool based on lattice theory focused on the processing of information from datasets, in which two sets (a set of objects and a set of attributes) and the relationships among them stand out. This interpretation of the dataset is called context. FCA has already been  used in different frameworks, such as in machine learning,  digital forensic~\cite{Aragon2023,OJEDAHERNANDEZ2024301797,SOKOL2023108940}, collaboration strategy~\cite{FELDE2023102104}, concept-drift research~\cite{FENZA2023122640}, medical diagnoses~\cite{biomimetics9070421,Zhai24},  social networks~\cite{pablo2015,sci:kristina19}, etc.

The different theoretical developments have contribute to offer as much flexible frameworks as possible in order to be able to handle different kinds of datasets~\cite{ANTONI2024109245,DUBOIS2025121419,KRIDLO2023119498,OJEDAHERNANDEZ2023108458}. 
For example, FCA is 
one of the most interesting tools to complement and give an high level of trustworthy to the nowadays neural networks decision support systems approaches.

It is very useful that the datasets contain imperfect, uncertainty or incomplete data. This was the main reason why Burusco and Fuentes-Gonz\'alez introduced the first fuzzy approach of FCA. Fuzzy logic was introduced by Lotfy A. Zadeh~\cite{Z1}
with this main goal, processing of inaccurate datasets. 
Different FCA fuzzy extension have been introduced, such as the B{\v e}lohl{\'a}vek's~\cite{Belohlavek99}, Pollandt's~\cite{pollandt97} and Kraj{\v c}i's~\cite{Stano:GCL} approaches. 
Later, other more flexible frameworks were introduced such as the multi-adjoint and the heterogeneous approaches~\cite{Antoni2016,krajci2013}.
This paper will consider the multi-adjoint concept lattice framework, which allows the use of a general algebraic structure, taking into account, for instance, diverse adjoint triples. 
{In fact, it allows to consider different  degrees of preference among the objects and/or attributes, see~\cite{cornejoINS650,mor-fss-cmpi} for more details. 
Moreover, this framework enables us to model effectively in the absence of properties such as commutativity and associativity, which could be required in  certain real-life scenarios.}
A  directly adaptation of the results given in this paper can be done to the heterogeneous approach.

Datasets are usually also very huge and obtaining information from them is very complex in many cases. This is the main reason why a large number of approaches are focused on the decomposition of datasets and the procurement of distributive systems.
Recent works in FCA~\cite{ARAGON2025-FSS,aragonESCIM2022,COAMfac2024} have introduced preliminary definitions and results in this fundamental goal. 
For example,  
one of the main goals of \cite{COAMfac2024} was to analyze  different properties of factorizable classical contexts, and study what of them can be extended to the fuzzy setting.

This paper is focused on the characterization of the detection and computation of the independent subcontexts appearing in a context~\cite{COAMfac2024,dubois:2012}. 
In~\cite{dubois:2012} the subcontexts were characterized by the necessity operators given in possibility theory~\cite{dubois:2011}. Specifically, the closures of subsets of objects and subsets of attributes, given by the composition of two necessity operators, provides the independent subcontexts.   However, this mechanism cannot trivially be extended to the fuzzy setting. This paper will characterize the pairs of fuzzy subsets of objects and subsets of attributes determining the independent fuzzy subcontexts of a given context.
Furthermore, different interesting properties of these operators will be analyzed, such as, each pair provides the top and bottom concepts of the corresponding independent subcontext and that no other concepts exists between them and the top and bottom concepts of the whole concept lattice of the original context.  

The structure of this paper is as follows:
Section~\ref{sec:preli} outlines the fundamental notions about the multi-adjoint framework that are essential for a more comprehensive understanding of this work. Section~\ref{sec:3} presents a detailed analysis of the closure of the necessity operator, which is employed to decompose a context into independent subcontexts. Section~\ref{sec:4} extends the properties  introduced in~\cite{COAMfac2024}, and a method based on thresholds to decompose a given context is provided in Section~\ref{sec:5}. Finally, conclusions and future perspectives are stated in Section~\ref{sec:conclu}.

\section{Preliminaries}\label{sec:preli}
This section presents the basic definitions and results needed throughout the paper. The first definition recalls the operators considered in the different algebraic structures used in this paper.

\begin{definition}\label{def:adjoint}
Let $(P_1,\leq_1)$, $(P_2,\leq_2)$, $(P_3,\leq_3)$ be posets and
$\adjoint\colon P_1\times P_2 \to P_3$,
$\dimp\colon P_3 \times P_2 \to P_1$,
$\uimp\colon P_3 \times P_1 \to P_2$ be
mappings, then $(\adjoint,\dimp,\uimp)$  is an
\emph{adjoint triple} with respect to $P_1, P_2, P_3$ if:
\begin{equation}\label{ap}
x\leq_1 z  \dimp y \quad \!\!\hbox{iff} \!\!\quad x \adjoint y \leq_3 z    \quad\!\!  \hbox{iff} \!\!\quad  y\leq_2 z \uimp  x
\end{equation}
  where $x\in P_1$, $y\in P_2$ and $z\in P_3$. Condition~\eqref{ap}  is called \emph{adjoint property}.
\end{definition}
 
In the following example, we show the adjoint triples associated with the G\"odel and product t-norms together with their residuated implications~\cite{comparativeAMIS15}, which will be employed in the forthcoming examples of this paper.  

\begin{example}\label{ex.non.commut}
Given   $m\in \mathbb N$, the set $[0,1]_m$ is   a regular partition of $[0,1]$ in $m$ pieces, that is, for example  $[0,1]_5=\{0,0.2,0.4,0.6,0.8,1\}$ divides the unit interval in five pieces.
A discretization of the G\"odel and product t-norms are the operators  $\adjoint^*_\G,\adjoint^*_{\Pp}\colon [0,1]_{5}\times [0,1]_{4}\to [0,1]_{10}$ defined, respectively,  as:
\[
x \adjoint^*_{\G} y= \frac{\textstyle  \lceil 10\cdot \min\{x,y\}\rceil}{\textstyle 10} \qquad
x \adjoint^*_{\Pp} y= \frac{\textstyle  \lceil 10\cdot x\cdot y\rceil}{\textstyle 10}
\]
for all $x\in [0,1]_{5}$ and $y\in[0,1]_{4}$, where $\lceil\,\_\,\rceil$ is the ceiling function. Moreover, their respective residuated implications
 $\swarrow^*_{\G},\swarrow^*_{\Pp}  \colon
[0,1]_{10} \times [0,1]_{4}\to [0,1]_{5}$ and 
$\nwarrow^*_{\G}, \nwarrow^*_{\Pp}  \colon [0,1]_{10} \times [0,1]_{5}\to [0,1]_{4}$ are defined as:
\begin{eqnarray*}
z\swarrow^*_{\G} y  & =  \dfrac { \lfloor5 \cdot z\ot_\G y\rfloor}{5}\qquad
z\swarrow^*_{\Pp} y  & =  \dfrac { \lfloor 5 \cdot z\ot_\Pp y\rfloor}{5}  \\
z\nwarrow^*_{\G} x  & = \dfrac { \lfloor 4 \cdot  z\ot_\G x\rfloor}{4}\qquad
z\nwarrow^*_{\Pp} x  & =  \dfrac { \lfloor 4 \cdot  z\ot_\Pp x\rfloor}{4}  \\
\end{eqnarray*}
where $\lfloor\,\_\,\rfloor$ is the floor function and the implications $ \ot_\G, \ot_\Pp\colon [0,1]\times [0,1]\to [0,1]$ are the residuated implications of the G\"odel and product t-norms, respectively, defined as:
\begin{eqnarray*}
b\ot_\G a &=
\begin{cases}
 1& \hbox{if } a\leq b\\
 b &\hbox{otherwise}
\end{cases} \qquad 
b\ot_\Pp a &=
\begin{cases}
 1& \hbox{if } a\leq b\\
 \frac{b}{a} &\hbox{otherwise}
\end{cases} 
\end{eqnarray*}
for all $a,b\in[0,1]$.  \qed
\end{example}

Note that, for example, the G\"odel and product discretization operators do not have zero-divisors. Recall that, given three lower bounded posets,  $(P_1,\leq_1,\bot_1)$, $(P_2,\leq_2,\bot_2)$, $(P_3,\leq_3,\bot_3)$,   an operator $\adjoint\colon P_1\times P_2\to P_3$ has \emph{zero-divisors}, if there exist at least two elements $x\in P_1\setminus\{\bot_1\}$ and $y\in P_2\setminus\{\bot_2\}$, such that $x\adjoint y=\bot_3$.

The following result states some properties of adjoint triples that will be used  later on.

\begin{proposition}[\cite{ija-cmr15}]\label{prop:atproperties}
Let $(\adjoint,\swarrow,\nwarrow)$ be an adjoint triple with respect to   
 three posets $(P_1,\leq_1)$, $(P_2,\leq_2)$ and $(P_3,\leq_3)$. The following properties are satisfied:
\begin{enumerate}
\item $\adjoint$ is order-preserving on both arguments.
\item $\swarrow$ and $\nwarrow$ are order-preserving on the first argument and order-reversing on the second argument.
\item $\bot_1\adjoint y=\bot_3$,  $\top_3\swarrow y=\top_1$, for all $y\in P_2$, when $(P_1,\leq_1,\bot_1,\top_1)$ and $(P_3,\leq_3,\bot_3,\top_3)$ are bounded posets.
\item $x\adjoint \bot_2=\bot_3$ and $\top_3\nwarrow x=\top_2$, for all $x\in P_1$, when $(P_2,\leq_2,\bot_2,\top_2)$ and $(P_3,\leq_3,\bot_3,\top_3)$ are bounded posets.
\item $z\nwarrow\bot_1=\top_2$ and $z\swarrow\bot_2=\top_1$, for all $z\in P_3$, when $(P_1,\leq_1,\bot_1,\top_1)$ and $(P_2,\leq_2,\bot_2,\top_2)$ are bounded posets.
\item $z\swarrow y =\max\{x\in P_1\mid x\adjoint y\leq_3 z\}$, for all $y\in P_2$ and $z\in P_3$.
\item $z\nwarrow x =\max\{y\in P_2\mid x\adjoint y\leq_3 z\}$, for all $x\in P_1$ and $z\in P_3$.
\end{enumerate}
\end{proposition}

 As a consequence, another helpful result about adjoint triples in related to the zero-divisors of an operator.
\begin{corollary}\label{cor:topno0}
    Given an adjoint triple $(\adjoint,\dimp,\uimp)$ with respect to three bounded posets $(P_1,\leq_1,\bot_1,\top_1), (P_2,\leq_2,\bot_2,\top_2)$ and $(P_3,\leq_3,\bot_3,\top_3)$, whose elements $T_1$ and $T_2$ are not zero-divisors of $\adjoint$, then it is satisfied that:
    \begin{itemize}
        \item $x\adjoint \top_2 = \bot_3$ if and only if $x = \bot_1$.
        \item $\top_1\adjoint y = \bot_3$ if and only if $y=\bot_2$.
    \end{itemize}
\end{corollary}
\begin{proof}
The proof is straightforwardly obtained from the definition of zero-divisors of an operator and Proposition~\ref{prop:atproperties}~(3).\qed
\end{proof}

In this paper, we will use three algebraic structures. The first one, which is recalled next, is the algebraic structure of the fuzzy extension of formal concept analysis from the multi-adjoint philosophy.   

\begin{definition}  
A \emph{multi-adjoint frame} is a tuple $(L_1,L_2,P, \adjoint_1, \dots,\adjoint_n)$,\break
where $(L_1,\preceq_1, \bot_1, \top_1)$ and $(L_2,\preceq_2, \bot_2, \top_2)$ are complete lattices, $(P,\leq)$ is a poset and $(\adjoint_i,\dimp^i,\uimp_i)$ is an adjoint triple with respect to $L_1, L_2, P$,  for all~$i\in\{1, \dots, n\}$.
\end{definition}

From a fixed multi-adjoint frame, a context  is defined as follows.

\begin{definition}
Given a multi-adjoint  frame $(L_1,L_2,P, \adjoint_1,\dots,\adjoint_n)$, a  \emph{context} is a tuple $(A,B,R,\sigma)$ such that
 $A$ and $B$ are  non-empty
sets (usually interpreted as attributes and objects, respectively), $R$ is a $P$-fuzzy relation $R \colon A\times B
\to P$ and $\sigma\colon A\times B\to \{1,\dots,n\}$ is a
mapping  which associates any element in $A\times B$ with some specific adjoint triple of the frame. 
\end{definition}

Given a lower bounded $(P,\leq,\bot)$, a context $(A,B,R,\sigma)$ will be called \emph{normalized} if for every attribute $a\in A$ there exist $b_1,b_2\in B$ such that $R(a,b_1)\neq \bot$ and $R(a,b_2)=\bot$ and for every object $b\in B$ there exist $a_1,a_2\in B$ such that $R(a_1,b)\neq \bot$ and $R(a_2,b)=\bot$.

In addition, from a multi-adjoint  frame $(L_1,L_2,P, \adjoint_1,\dots,\adjoint_n)$ and context $(A,B,R,\sigma)$, the generalization of the derivation operators are the mappings ${}^{\uparrow} \colon L_2^B\rightarrow L_1^A$ and ${}^{\downarrow} \colon L_1^A \rightarrow L_2^B$ given as follows:
 \begin{eqnarray*}\label{conexGmulti}
	 g^{\uparrow} (a) &=&\inf \{  R(a,b)\swarrow^{\sigma(a,b)}  g(b)\mid b\in B\}\\\label{conexGmulti2}
	 f^{\downarrow} (b)  &=&\inf \{  R(a,b)\nwarrow_{\sigma(a,b)} f(a)\mid a\in A\}
\end{eqnarray*}
for all $g\in L_2^B$, $f\in L_1^A$ and $a\in A$, $b\in B$.  
Moreover, the pair $({}^\up,{}^\down)$ forms an antitone Galois connection~\cite{mor-fss-cmpi}. A {\em multi-adjoint concept} is a pair $\langle g, f\rangle$ satisfying that $g^{\uparrow} =f$ and $f^{\downarrow} =g$, and  the set of multi-adjoint concepts with the ordering defined in the following definition forms a complete lattice.

\begin{definition}[\cite{mor-fss-cmpi}]
 The  \emph{multi-adjoint concept lattice} associated with a  multi-adjoint frame $(L_1,L_2,P, \adjoint_1, \dots,\adjoint_n  )$ and a context $(A,B,R,\sigma)$ given, is the set
$$
\mathcal{M}=\{\langle g, f\rangle \mid  g\in L_2^B, f\in L_1^A
\hbox{ and } g^{\uparrow} =f, f^{\downarrow}=g\}
$$
where the
ordering is defined by $ \langle g_1, f_1\rangle\preceq \langle g_2, f_2\rangle
\hbox{ if and only if } g_1\preceq_2 g_2 $ (equivalently $f_2\preceq_1
f_1 $), for all $\langle g_1, f_1\rangle, \langle g_2, f_2\rangle\in \mathcal{M}$.
\end{definition}

The other two algebraic structures are associated with the hybrid concept lattice frameworks obtained from the merging of  formal concept analysis and rough set theory, that is, the property-oriented concept lattice and object-oriented concept lattice~\cite{ins-medina}. The following definitions are the corresponding ones in the multi-adjoint setting.

\begin{definition}\label{def:pof}
A \emph{multi-adjoint property-oriented frame} is the tuple \break$(L_1,L_2,P, \adjoint^\po_1,
\dots,\adjoint^\po_n  )$, where $(\adjoint^\po_i,\swarrow_\po^i,\nwarrow^\po_i)$ is an adjoint triple with respect to $P$, $L_2$, $L_1$ for all~$i\in\{1, \dots, n\}$.
\end{definition}

 In this frame, a context is the tuple $(A,B,R,\sigma_\po)$, where $A$, $B$ and $R$ are as in the case above, and $\sigma_\po$ associates each pair $(a,b)\in A\times B$ with a triple of the frame $(L_1,L_2,P, \adjoint^\po_1,
\dots,\adjoint^\po_n  )$. The necessity operator is given by the mapping ${}^{\dN} \colon L_1^A \rightarrow L_2^B$, defined as:
$$
f^\dN (b)  =\inf \{  f(a)\nwarrow_{\sigma_\po(a,b)} R(a,b)\mid a\in A\}
$$
for all $b\in B$ and $f\in L_1^A $.

\begin{definition}\label{def:oof}
A \emph{multi-adjoint object-oriented frame} is the tuple \break$(L_1,L_2,P, \adjoint^\oo_1,\dots,\adjoint^\oo_n  )$, where $(\adjoint^\oo_i,\swarrow_\oo^i,\nwarrow^\oo_i)$ is an adjoint triple with respect to $L_1$, $P$, $L_2$ for all~$i\in\{1, \dots, n\}$. 
\end{definition}

A context $(A,B,R,\sigma_\oo)$ is defined similarly and the necessity operator is given by the mapping ${}^{\uN} \colon L_2^B\rightarrow L_1^A$, defined as:
$$
g^{\uN} (a) =\inf \{  g(b)\swarrow^{\sigma_\oo(a,b)}  R(a,b)\mid b\in B\}
$$
for all $a\in A$ and $g\in L_2^B $. Notice that these necessity operators are the generalization of the necessity operator defined in the classical setting~\cite{chenyao08,dubois:2012,DuntschG03}. Moreover, we can see that the implications are defined in  different domains.

In addition, the fuzzy sets $g\in L_2^B $ and $f\in L_1^A $ such that $g(b) = \top_2$, for all $b\in B$, and $f(a)=\top_1$, for all $a\in A$, are denoted as $g_\top$ and $f_\top$, respectively. Similarly,  when  $g(b) = \bot_2$, for all $b\in B$, and $f(a)=\bot_1$, for all $a\in A$, we will denote them as $g_\bot$ and $f_\bot$, respectively.

This section finishes presenting  two fundamental notions in this paper. The first one fixes the preliminary properties that a subcontext must satisfy in order to obtain an independent subcontext. {Notice that these definitions extend the ones given in the classical case~\cite{dubois:2012} based on the bottom element and taking into account all the non-zero relations. This last fact will be weakened with the consideration of a threshold in Section~\ref{sec:5}.}

\begin{definition}\label{def.Lsep_subc}

Given the multi-adjoint frame $(L_1,L_2,P, \adjoint_1, \dots,\adjoint_n)$ and a context $(A,B,R,\sigma)$, a \emph{separable subcontext} is a tuple\footnote{Notice that $R_{Y\times X}$ and $\sigma_{Y\times X}$ denote the restriction of the relation $R$ and the mapping $\sigma$ to the Cartesian product $Y\times X$.} $(Y,X,R_{Y\times X},\sigma_{Y\times X})$ such that 
\begin{itemize}
    \item $Y\subset A$ and $X\subset B$ are  non-empty sets.
    \item There exist $a\in Y$ and $b\in X$ such that $R(a,b)\neq \bot$.
    \item $R(a,b')=\bot$, for all $(a,b')\in Y\times X^c$.
    \item $R(a',b)=\bot$, for all $(a',b)\in Y^c\times X$.
\end{itemize}
where ${}^c$ denotes the complement of a set.
\end{definition}

Based on the previous notion, we can say when a context can be decomposed into independent subcontexts.

\begin{definition}\label{def:decomp_indep}

A normalized context $(A,B,R,\sigma)$ has a \emph{decomposition  into independent subcontexts},  
if there exists a non-empty index set $\Lambda$ such that:
\begin{itemize}
 \item $(A_\lambda, B_\lambda, R_{A_\lambda\times B_\lambda}, \sigma_{A_\lambda\times B_\lambda})$ is a separable subcontext of  $(A,B,R,\sigma)$, for all ${\lambda\in\Lambda}$.
	\item $\bigcup_{\lambda \in \Lambda} A_\lambda = A$,  $\bigcup_{\lambda \in \Lambda} B_\lambda = B$, and  $A_\lambda\cap  A_\mu=\varnothing$, $B_\lambda\cap  B_\mu=\varnothing$, for all $\lambda,\mu\in \Lambda$ with $\lambda\neq\mu$.
        \item The mapping $\sigma$ associates conjunctors with no zero-divisor for the subsets $A_\lambda^c\times B_\lambda$ and $A_\lambda\times B_\lambda^c$ of $A\times B$, for all $\lambda\in\Lambda$.
\end{itemize} 
 Every tuple $(A_\lambda, B_\lambda, R_{A_\lambda\times B_\lambda}, \sigma_{A_\lambda\times B_\lambda})$  is called \emph{independent subcontext of the context $(A,B,R,\sigma)$.}
\end{definition}

\section{Closure of necessity operators to obtain independent subcontexts} \label{sec:3}
In this section, we study when a formal context within the multi-adjoint framework can be decomposed into independent subcontexts and how  these independent subcontexts can be determined. To this end, 
necessity operators will be fundamental, as in the classical case~\cite{COAMfac2024, dubois:2012}.

It is important to note that, 
in the fuzzy setting, the truth-value algebraic structure is determinant to define the necessity operators. Moreover, these operators belong to two related but different frameworks given by the property-oriented and the object-oriented concept lattices point of views. Specifically, although it is natural to fix the same   set of attributes, set of objects and the $P$-fuzzy relation, the operators (implications/adjoint triples) are defined from the two different frameworks aforementioned.

The following notation will be established on the applications $\sigma$ and the adjoint triples of each frame in order to facilitate the identification of the framework in which we are working. We will write $(A,B,R,\sigma)$ as a context associated with $(L_1, L_2, P,\adjoint_1,\dots, \adjoint_n)$ which is a multi-adjoint frame, $(A,B,R,\sigma_\po)$ as a context associated with $(L_1, L_2, P,\adjoint^\po_1,\dots, \adjoint^\po_m)$ which is a multi-adjoint property-oriented frame  and $(A,B,R,\sigma_\oo)$ as a context associated with $(L_1, L_2, P,\adjoint^\oo_1,\dots, \adjoint^\oo_s)$ which is a multi-adjoint object-oriented frame.

From now on, we will fix a multi-adjoint frame $(L_1, L_2, P,\adjoint_1,\dots, \adjoint_n)$, a multi-adjoint property-oriented frame $(L_1, L_2, P,\adjoint^\po_1,\dots, \adjoint^\po_m)$ and a multi-adjoint object-oriented frame $(L_1, L_2, P,\adjoint^\oo_1,\dots, \adjoint^\oo_s)$, where $(P,\leq, \bot,\top)$ is a bounded poset, and the conjunctors $\adjoint_i$, $\adjoint^\po_j$ and $\adjoint^\oo_k$ have no zero-divisors, for all $i\in \{1,\dots, n\}$, $j\in \{1,\dots, m\}$ and $k\in \{1,\dots, s\}$. Additionally, we will only mention the context associated with the multi-joint framework, although we can use the mappings $\sigma_\po$ and $\sigma_\oo$ when applying the necessity operators.

The following result is a technical result that will be useful to prove further results of this section.
\begin{lemma}\label{lem:botrelation}
	Let $(A,B,R,\sigma)$ be a context,  and $(g,f)$ a pair of fuzzy subsets and $(a,b)\in A\times B$.
	\begin{itemize}
		\item If $g = f^{\dN}$,  $f(a) = \bot_1$ and $g(b) = \top_2$, then $R(a,b) = \bot$.
		\item If $f = g^{\uN}$,  $f(a) = \top_1$ and $g(b) = \bot_2$, then $R(a,b) = \bot$.
	\end{itemize}
\end{lemma}

\begin{proof}
	We consider $(a,b)\in A\times B$ such that $g = f^{\dN}$, $f(a) = \bot_1$ and $g(b) = \top_2$. In this case, we have that \[\top_2 = g(b) = f^{\dN}(b) = \inf\{f(a')\uimp_{\sigma_\po} R(a',b) \mid a'\in A\}\]
	Therefore, $f(a')\uimp_{\sigma_\po} R(a',b)= \top_2$, for all $a'\in A$. In particular,  this equality holds for $a$, that is, $f(a)\uimp_{\sigma_\po} R(a,b)= \top_2$. Hence, since $f(a) = \bot_1$, we have that $\bot_1\uimp_{\sigma_\po} R(a,b)= \top_2$ and, since $\top_2$ is not a zero-divisor by Proposition~\ref{prop:atproperties}~(7) and Corollary~\ref{cor:topno0}, we conclude that $R(a,b)=\bot$.
	
	The second condition arises analogously.\qed
\end{proof}

We will  make use of a particular Boolean context which is defined from the $P$-fuzzy relation of the context.
\begin{definition}
Given a context $(A,B,R,\sigma)$ associated with a multi-adjoint frame $(L_1, L_2, P,\adjoint_1,\dots, \adjoint_n)$, where $(P,\leq, \bot,\top)$ is a bounded poset,  the  Boolean relation $R^\B\colon A\times B\to \{0,1\}$ can be defined as follows:
\[R^{\B}(a,b) = \left\{\begin{array}{ll}
	1 & \mbox{ if }   R(a,b)\neq \bot \\
	0 & \mbox{ otherwise}  
\end{array}\right.\]
The context $(A,B,R^\B)$ is called \emph{associated Boolean context of the context $(A,B,R,\sigma)$}.
\end{definition}

From now on, we will consider a normalized context $(A,B,R,\sigma)$ and its associated Boolean context $(A, B, R^\B)$ throughout the document.
The following result is also a technical result that will play a key role to determine independent subcontexts of a formal context. We will abuse the notation and we will use the same symbol to denote  the necessity operator on crisp sets as on fuzzy subsets.

\begin{lemma}\label{lem:necessity}
Given $X\subseteq B$ and $Y\subseteq A$, the following equalities hold:
    $$
    \bigchi_X^\uN = \bigchi_{X^\uN} \quad \mbox{and}\quad \bigchi_Y^\dN = \bigchi_{Y^\dN}
    $$
where $\bigchi_X\colon B\to \{\bot_{2},\top_2\}$  and $\bigchi_Y\colon A\to \{\bot_{1},\top_1\}$ are the characteristic functions of the sets $X$ and $Y$, respectively.
\end{lemma}
\begin{proof}
Let us consider any attribute $a\in A$ and a subset of objects $X\subseteq B$. By the definition of the necessity operator we have that
$$
X^\uN = \{ a\in A \mid \mbox{for each }b\in B, \mbox{ if } R^\B(a,b)=1, \mbox{ then } b\in X\}
$$
Now, we can distinguish two cases:
\begin{itemize}
    \item If $a\not\in X^\uN$, then $\bigchi_{X^\uN}(a) = \bot_1$ and there is  $b'\in B$ such that $R^\B(a,b') = 1$ and $b'\not\in X$. In particular, we have that $R(a,b')\neq \bot$ and the implication $\bot_2\dimp^{\sigma_\oo} R(a,b') =\bot_1$ since $\adjoint^\oo_k$ has no zero-divisors, for all $k\in\{1,\dots,s\}$. 
    Therefore, we have that 
    \begin{align*}
        \bigchi_X^\uN(a) &= \inf\{\bigchi_X(b)\dimp^{\sigma_\oo} R(a,b)\mid b\in B\}\\
        &\preceq_1  \bigchi_X(b')\dimp^{\sigma_\oo} R(a,b')\\
        &= \bot_2\dimp^{\sigma_\oo} R(a,b')\\
        &= \bot_1
    \end{align*}
for all $a\not\in X^\uN$. Thus, the equality $\bigchi_X^\uN(a) = \bigchi_{X^\uN} (a)$ holds, for all $a\not\in X^\uN$.
\item If $a\in X^\uN$, then we have that $\bigchi_{X^\uN}(a) = \top_1$ and, for all $b\in B$ we have either $R^\B(a,b) = 0$ or, $R^\B(a,b) = 1$ and $b\in X$. Now, to compute
$\bigchi_X^\uN(a) = \inf\{\bigchi_X(b)\dimp^{\sigma_\oo} R(a,b)\mid b\in B\}$ we have that
\begin{itemize}
    \item if $R^\B(a,b') = 0$, then $R(a,b') = \bot$, for some $b'\in B$. Therefore, by Proposition~\ref{prop:atproperties}~(5), $\bigchi_X(b)\dimp^{\sigma_\oo} R(a,b') = \bigchi_X(b)\dimp^{\sigma_\oo} \bot = \top_1$.

    \item if $R^\B(a,b') = 1$ and $b'\in X$, 
    we have by Proposition~\ref{prop:atproperties}~(3) that
    $\bigchi_X(b')\dimp^{\sigma_\oo} R(a,b') = \top_2 \dimp^{\sigma_\oo} R(a,b') = \top_1$. 
\end{itemize}
Therefore, we have that $\bigchi_X^\uN(a) =\top_1$, for all $a\in X^\uN$ and, as a consequence, the equality $\bigchi_X^\uN(a) = \bigchi_{X^\uN} (a)$ holds, for all $a\in X^\uN$.
\end{itemize}

Thus, from the previous developments we obtain that $\bigchi_X^\uN = \bigchi_{X^\uN}$. The proof of the equality $\bigchi_Y^\dN = \bigchi_{Y^\dN}$ follows analogously. \qed
\end{proof}

Furthermore, it is necessary to recall that from a context $(A,B,R,\sigma)$ and its associated Boolean context $(A,B,R^\B)$, we can consider the following two sets, which were defined in~\cite{COAMfac2024}.
\begin{align*}
    \cC_N &= \{(X,Y)\mid X\subseteq A, Y\subseteq B \mbox{ and } X^\uN = Y, X = Y^\dN\}\\
    \cF_N &= \{(g,f)\mid g\in L_2^B,   f\in L_1^A \mbox{ and } g^\uN = f, g = f^\dN\}
\end{align*}

\begin{remark}\label{re:dp}
Dubois and Prade shown  in~\cite{dubois:2012} that each pair  belonging to $\mathcal C_N$ determines an independent subcontext   of the original context. Specifically, given a    context $(A,B,R)$ and a  pair $(X, Y )\in\mathcal C_N$, we have that  $(Y,X, R_{Y\times X})$  is  an independent subcontext  of $(A,B,R)$.
\end{remark}
The following example will serve to illustrate the aforementioned sets within a given context.

\begin{example}\label{ex:FN_CN}
 Let us consider the multi-adjoint frame $(L_1,L_2,P,\adjoint^*_\G,\adjoint^*_\Pp)$, the property-oriented multi-adjoint frame $(L_1,L_2,P,\adjoint^*_\G)$ and the object-oriented multi-adjoint frame $(L_1,L_2,P,\adjoint^*_\Pp)$, where $L_1 = [0,1]_5, L_2 = [0,1]_4$ and $P = [0,1]_{10}$ represent partitions of the unit interval in 5, 4 and 10 pieces, respectively. Moreover, $\adjoint_\G^*$ and $\adjoint_\Pp^*$ are the discretization of the G\"odel and product t-norms, respectively~\cite{ija-cmr15,caepia-bires}. Now, we consider the context $(A,B,R,\sigma)$ where $A=\{a_1,a_2,a_3,a_4\}$, $B = \{b_1,b_2,b_3,b_4\}$ and the fuzzy relation and the mapping $\sigma$ are given in Table~\ref{tab:contexto0}.

    \begin{table}[ht]
        \centering
        \begin{minipage}{0.45\textwidth}
            \begin{tabular}{|c|c c c c|}
            \hline
             $R$ & $b_1$ & $b_2$ & $b_3$ & $b_4$  \\
             \hline
             $a_1$ & 1.0 & 0.0 & 0.6 & 0.0 \\
             $a_2$ & 0.7 & 0.0 & 0.8 & 0.0 \\
             $a_3$ & 0.0 & 0.0 & 0.0 & 0.3 \\
             $a_4$ & 0.0 & 0.5 & 0.0 & 0.0 \\
             \hline
        \end{tabular}
        \end{minipage}
        \begin{minipage}{0.45\textwidth}
            \begin{tabular}{|c|c c c c|}
                     \hline
                     $\sigma$ & $b_1$ & $b_2$ & $b_3$ & $b_4$  \\
                     \hline
                     $a_1$ & $\adjoint^*_\G$ & $\adjoint^*_\Pp$ & $\adjoint^*_\G$ & $\adjoint^*_\Pp$ \\
                     $a_2$ & $\adjoint^*_\Pp$ & $\adjoint^*_\G$ & $\adjoint^*_\Pp$ & $\adjoint^*_\G$ \\
                     $a_3$ & $\adjoint^*_\G$ & $\adjoint^*_\Pp$ & $\adjoint^*_\G$ & $\adjoint^*_\Pp$ \\
                     $a_4$ & $\adjoint^*_\Pp$ & $\adjoint^*_\G$ & $\adjoint^*_\Pp$ & $\adjoint^*_\G$ \\
                     \hline
                \end{tabular}
        \end{minipage}
                
        \caption{Fuzzy relation $R$ and the mapping $\sigma$ of the context $(A,B,R,\sigma)$ of Example~\ref{ex:FN_CN}.}
        \label{tab:contexto0}
    \end{table}
In addition, the mappings $\sigma_p$ and $\sigma_o$ are constant since each frame only has one triple. Now, we can provided the set $\cF_N$ whose elements are listed below
\begin{align*}
(g_1,f_1) = & ( \{\} ,  \{\} )\\
(g_2,f_2) = & ( \{b_4/1.0\}, \{a_3/1.0\}  )\\
(g_3,f_3) = & ( \{b_2/0.25\}, \{a_4/0.4\} )\\
(g_4,f_4) = & ( \{b_2/0.25, b_4/1.0\} , \{a_3/1.0, a_4/0.4\})\\
(g_5,f_5) = & (\{b_2/1.0\} ,  \{a_4/1.0\})\\
(g_6,f_6) = & ( \{b_2/1.0, b_4/1.0\}, \{a_3/1.0, a_4/1.0\} )\\
(g_7,f_7) = & ( \{b_1/1.0, b_3/1.0\}, \{a_1/1.0, a_2/1.0\} )\\
(g_8,f_8) = & ( \{b_1/1.0, b_3/1.0, b_4/1.0\},  \{a_1/1.0, a_2/1.0, a_3/1.0\})\\
(g_9,f_9) = & (\{b_1/1.0, b_2/0.25, b_3/1.0\}, \{a_1/1.0, a_2/1.0, a_4/0.4\})\\
(g_{10},f_{10}) = & (\{b_1/1.0, b_2/0.25, b_3/1.0, b_4/1.0\}, \{a_1/1.0, a_2/1.0, a_3/1.0, a_4/0.4\} )\\
(g_{11},f_{11}) = & (\{b_1/1.0, b_2/1.0, b_3/1.0\}, \{a_1/1.0, a_2/1.0, a_4/1.0\})\\
(g_{12},f_{12}) = & (\{b_1/1.0, b_2/1.0, b_3/1.0, b_4/1.0\}, \{a_1/1.0, a_2/1.0, a_3/1.0, a_4/1.0\} )
\end{align*}

It should be noted that the attributes and objects with a value of $0.0$ are omitted in order to facilitate a more comprehensive understanding of the concepts when writing fuzzy sets.
Furthermore, we can obtain the associated Boolean context of $(A,B,R,\sigma)$, that is, $(A,B,R^\B)$ which is depicted in Table~\ref{tab:contexto0_B}.
    \begin{table}[ht]
        \centering
        \begin{tabular}{|c|c c c c|}
        \hline
         $R^\B$ & $b_1$ & $b_2$ & $b_3$ & $b_4$ \\
         \hline
         $a_1$ & 1 & 0 & 1 & 0 \\
         $a_2$ & 1 & 0 & 1 & 0 \\
         $a_3$ & 0 & 0 & 0 & 1 \\
         $a_4$ & 0 & 1 & 0 & 0 \\
         \hline
        \end{tabular}
        \caption{Boolean relation $R^B$ of the associated Boolean context of the context $(A,B,R,\sigma)$ in Example~\ref{ex:FN_CN}.}
        \label{tab:contexto0_B}
    \end{table}
From this associated Boolean context, we obtain the elements of the set $\cC_N$ which are listed below

\begin{minipage}{0.45\textwidth}
    \begin{align*}
    (X_1,Y_1) = & ( \varnothing ,  \varnothing )\\
    (X_2,Y_2) = & ( \{a_3\} ,  \{b_4\} )\\
    (X_3,Y_3) = & ( \{a_4\} ,  \{b_2\} )\\
    (X_4,Y_4) = & ( \{a_3, a_4\} ,  \{b_2, b_4\} )
    \end{align*}
\end{minipage}
\begin{minipage}{0.45\textwidth}
    \begin{align*}
    (X_5,Y_5) = & ( \{a_1,a_2\} ,  \{b_1,b_3\} )\\
    (X_6,Y_6) = & ( \{a_1,a_2,a_3\} ,  \{b_1,b_3,b_4\} )\\
    (X_7,Y_7) = & ( \{a_1, a_2, a_4\} ,  \{b_1, b_2, b_3\} )\\
    (X_8,Y_8) = & ( A ,  B )
    \end{align*}
\end{minipage}

\qed
\end{example}

The following result shows the existence of a closely relationship between the pairs in $\mathcal{F}_N$ obtained from a context $(A,B,R,\sigma)$, and the pairs in $\mathcal{C}_N$ obtained from its associated Boolean context $(A, B, R^{\B})$.

\begin{theorem}\label{th:fncn}
Given $X\subseteq B$ and $Y\subseteq A$, we have that $(\bigchi_X,\bigchi_Y)\in \cF_N$ if and only if  $(X,Y)\in\cC_N$.
\end{theorem}
\begin{proof}
    Let us suppose that $(\bigchi_X,\bigchi_Y)\in \cF_N$. Hence,  it is satisfied that   $\bigchi_X^{\uN} = \bigchi_Y$. By Lemma~\ref{lem:necessity}, we have that $\bigchi_Y =\bigchi_X^{\uN}= \bigchi_{X^{\uN}}$, and therefore, we can conclude that $X^{\uN} = Y$. Analogously, from $\bigchi_X = \bigchi_Y^{\dN}$ we obtain that $Y^{\dN} = X$. Thus, $(X,Y)\in \cC_N$.
    
    Let us prove the other implication. We consider any pair $(X,Y)\in \mathcal{C}_N$. Therefore, we have that $X^\uN = Y$ and, by Lemma~\ref{lem:necessity}, we obtain that $\bigchi_Y = \bigchi_{X^\uN} = \bigchi_X^\uN$. Moreover, from $X = Y^\dN$ we obtain that $\bigchi_X = \bigchi_{Y^\dN} = \bigchi_Y^\dN$. Consequently, $(\bigchi_X,\bigchi_Y)\in \cF_N$.\qed
\end{proof}

We illustrate this result in the next example as a continuation of Example~\ref{ex:FN_CN}.

\begin{example}\label{ex:FN_CN_th}

Coming back to Example~\ref{ex:FN_CN}, it is clear the relationship between the elements of the sets $\cF_N$ and $\cC_N$. There is an injective correspondence between the set $\cC_N$ with the set $\cF_N$, that is, 

\begin{minipage}{0.45\textwidth}
    \begin{align*}
    (\bigchi_{X_1},\bigchi_{Y_1}) = & ( g_1 , f_1 )\\
    (\bigchi_{X_2},\bigchi_{Y_2}) = & ( g_2 ,  f_2 )\\
    (\bigchi_{X_3},\bigchi_{Y_3}) = & ( g_5 ,  f_5 )\\
    (\bigchi_{X_4},\bigchi_{Y_4}) = & ( g_6 ,  f_6 )
    \end{align*}
\end{minipage}
\begin{minipage}{0.45\textwidth}
    \begin{align*}
    (\bigchi_{X_5},\bigchi_{Y_5}) = & ( g_7, f_7 )\\
    (\bigchi_{X_6},\bigchi_{Y_6}) = & ( g_8 ,  f_8 )\\
    (\bigchi_{X_7},\bigchi_{Y_7}) = & ( g_{11} ,  f_{11} )\\
    (\bigchi_{X_8},\bigchi_{Y_8}) = & ( g_{12} , f_{12} )
    \end{align*}
\end{minipage}

\qed
\end{example}

The aforementioned theorem has direct implications for the pairs $(g_\bot,f_\bot)$ and $(g_\top, f_\top)$ as we state below.
\begin{corollary}\label{cor:bottop}
    Given the set $\cF_N$, it is satisfied that $(g_\bot, f_\bot), (g_\top,f_\top)\in \mathcal F_N $.
\end{corollary}
\begin{proof}
    The proof straightforwardly holds from Theorem~\ref{th:fncn} given that the pairs $(\varnothing,\varnothing)$ and $(B,A)$ are elements of $\cC_N$. \qed
\end{proof}

On the one hand, from Theorem~\ref{th:fncn}, we can deduce that  when the associated Boolean context contains independent subcontexts, the cardinality of the set $\mathcal{F}_N$ will be greater than two. 
On the other hand, a pair $(g_i,f_i)$ of $\mathcal{F}_N$ that satisfies the conditions of Theorem~\ref{th:fncn}, with the exception of the pairs outlined in Corollary~\ref{cor:bottop}, determines disjoint partitions of  the set of objects and the set of attributes.  We can decompose the set of objects as $B=B_i^\bot\cup B_i^\top$, where $b\in B_i^\top$ if $g_i(b)=\top_2$  and  $b\in B_i^\bot$ if $g_i(b)=\bot_2$. Analogously, considering $f_i$, the set of attributes can be decomposed as $A= A_i^\bot\cup A_i^\top$, where the sets $A_i^\bot$ and $ A_i^\top$ are defined analogously to $B_i^\bot$ and $ B_i^\top$, respectively.
 
Hereafter, the pairs  of $\mathcal{F}_N$  satisfying the conditions of Theorem~\ref{th:fncn}, with the exception of those specified in Corollary~\ref{cor:bottop}, will be denoted by $\cF\cC$, that is, $\cF\cC\subseteq \cF_N$ is the set defined as:
$$
\cF\cC = \{(\bigchi_X,\bigchi_Y)\in\cF_N\mid \varnothing \neq X\subset B, \varnothing\neq Y\subset A \}
$$

Notice that, by Theorem~\ref{th:fncn}, the pairs $(X, Y)$ obtained from the pairs $(\bigchi_X,\bigchi_Y)\in\cF\cC$ belong to $\cC_N$. In addition, 
from the pairs of $\cF\cC$ disjoint partitions of the set of objects and the set of attributes can be defined, as the following result shows.

\begin{lemma}[\cite{COAMfac2024}]\label{lemma:complement}
Given a context $(A,B,R)$ and a pair $(X,Y)\in \mathcal{C}_N$, the complement of the pair $(X,Y)$ also belongs to $\mathcal{C}_N$, that is, $(X, Y)^c=(X^{c}, Y^{c})\in \mathcal{C}_N$.
\end{lemma}

The following result guarantees that when the set $\cF\cC$ is non-empty then it is possible to find partitions of the sets of attributes and objects. 

\begin{proposition}\label{prop:partition}
If  $\cF\cC\neq \varnothing$, then there exists a family of  pairs  $\{(g_i,f_i)\}_{i\in I}\subseteq \cF\cC$   such that
\begin{itemize}
\item $B = \bigcup_{i\in I}B_i^\top$, with $B_i^\top\cap B_j^\top=\varnothing$, for all $i,j\in I$ with $i\neq j$.
\item  $B = \bigcup_{i\in I}B_i^\bot$, with $B_i^\bot\cap B_j^\bot=\varnothing$, for all $i,j\in I$ with $i\neq j$.
\item $A = \bigcup_{i\in I}A_i^\top$, with $A_i^\top\cap A_j^\top=\varnothing$, for all $i,j\in I$ with $i\neq j$.
\item  $A = \bigcup_{i\in I}A_i^\bot$, with $A_i^\bot\cap A_j^\bot=\varnothing$, for all $i,j\in I$ with $i\neq j$.
\end{itemize}
\end{proposition}

\begin{proof}
 Let us assume that $\cF\cC\neq\varnothing$ and consider $(\bigchi_{X_1},\bigchi_{Y_1})\in\cF\cC$ with  $X_1\subset B$ and $Y_1\subset A$.
 Therefore, by Theorem~\ref{th:fncn}, we have that $(X_1,Y_1)\in \cC_N$, and by Lemma~\ref{lemma:complement}, $(B\setminus X_1,A\setminus Y_1)\in \cC_N$. Thus, considering $X_2 = B\setminus X_1$ and $Y_2 = A\setminus Y_1$, by Theorem~\ref{th:fncn}, we have that $(\bigchi_{X_2},\bigchi_{Y_2})\in \cF_N$, hence $(\bigchi_{X_2},\bigchi_{Y_2})\in\cF\cC$. 
 Now, considering the family of pairs $\{(\bigchi_{X_1},\bigchi_{Y_1}), (\bigchi_{X_2},\bigchi_{Y_2})\}\subseteq \cF\cC$, it is clear that $B_1^\top\cup B_2^\top = B$ with $B_1^\top\cap B_2^\top=\varnothing$, by the definition of these pairs.
 
 The proof for the rest of items follows analogously.\qed
\end{proof}

As a consequence of the proof above, we obtain the following result.

\begin{corollary}\label{cor:compl}
 Given  $(\bigchi_{X}, \bigchi_{Y})\in\cF\cC$, with $\varnothing \neq X\subset B, \varnothing\neq Y\subset A$,  we have that  $(\bigchi_{X^c}, \bigchi_{Y^c})\in\cF\cC$.
\end{corollary}

The previous results allows us to determine decompositions into independent subcontexts of a formal context within the multi-adjoint framework, as the following proposition proves.

\begin{proposition}\label{prop:FC_ind_subcontext}
The family of pairs given in Proposition~\ref{prop:partition} provides a decomposition into independent subcontexts of the context $(A, B, R,\sigma)$.
\end{proposition}
\begin{proof}
    Let us consider any pair $(\bigchi_{X}, \bigchi_{Y})\in\cF\cC$, with $\varnothing \neq X\subset B, \varnothing\neq Y\subset A$,  and by Theorem~\ref{th:fncn} we have that  $(X,Y)\in \cC_N$. Moreover, from Corollary~\ref{cor:compl}, we know that $\{(X,Y),(X^c,Y^c)\}$ is  a family of pairs of $\cC_N$ providing a partition of the sets of objects and attributes. 
    Now, we will show that the tuple $(Y,X,R_{Y\times X},\sigma_{Y\times X})$ is a separable subcontext. We know that $X\subset B$ and $Y\subset A$ are non-empty subsets. By Lemma~\ref{lem:botrelation} and $(\bigchi_{X}, \bigchi_{Y})\in\cF\cC$, we obtain that $R(a,b')=\bot$, for all $(a,b')\in Y\times X^c$ and $R(a',b)=\bot$, for all $(a',b)\in Y_1^c\times X$. Lastly, since the context is normalized, we can claim that there exist $a\in Y$ and $b\in X$ such that $R(a,b)\neq \bot$. Thus, $(Y,X,R_{Y\times X},\sigma_{Y\times X})$ is a separable subcontext. The proof to show that $(Y^c,X^c,R_{Y^c\times X^c},\sigma_{Y^c\times X^c})$ is a separable subcontext follows analogously. Therefore, since the conjunctors of the frame have no zero-divisors, we obtain that the context has a decomposition into independent subcontexts. \qed
\end{proof}
 
We will show the application of the aforementioned results in the next example.
 
\begin{example}\label{ex:family_fc}

We continue with Example~\ref{ex:FN_CN} and Example~\ref{ex:FN_CN_th}. In the latter, the elements showed in the example are indeed the elements of the set $\cF\cC$ except the pairs $(\bigchi_{X_1},\bigchi_{Y_1})$ and $(\bigchi_{X_8},\bigchi_{Y_8})$ since $(X_1,Y_1) = (\varnothing,\varnothing)$ and $(X_8, Y_8)=(B,A)$. From this set, since it is clearly not empty we can find a family of pair that provides a partition of the sets of attributes and objects as Proposition~\ref{prop:partition} states. All the different families that we can find satisfying Proposition~\ref{prop:partition} are listed below
\begin{align*}
    F_1 = &\{(\bigchi_{X_2},\bigchi_{Y_2}), (\bigchi_{X_3},\bigchi_{Y_3}), (\bigchi_{X_5},\bigchi_{Y_5})\}\\
    F_2 = &\{(\bigchi_{X_3},\bigchi_{Y_3}),(\bigchi_{X_6},\bigchi_{Y_6})\}\\
    F_3 = &\{(\bigchi_{X_4},\bigchi_{Y_4}),(\bigchi_{X_5},\bigchi_{Y_5})\}\\
    F_4 = &\{(\bigchi_{X_2},\bigchi_{Y_2}), (\bigchi_{X_7},\bigchi_{Y_7})\}
\end{align*}
 For each family, we can obtain a decomposition into independent subcontexts of the context $(A,B,R,\sigma)$. For instance, if we consider the family $F_2$, we have that the partition of the sets of attributes and objects can be provided by selecting the sets $A^\top$ and $B^\top$ of the corresponding pairs, that is, $A^\top_3 = \{a_4\}$ and $A^\top_6 = \{a_1,a_2,a_3\}$, and $B_3^\top = \{b_2\}$ and $B_6^\top = \{b_1,b_3,b_4\}$. 
 Therefore, a decomposition into independent subcontexts is $\{(A^\top_\lambda, B^\top_\lambda, R_{A^\top_\lambda\times B^\top_\lambda}, \sigma_{A^\top_\lambda\times B^\top_\lambda})\}_{\lambda\in\{3,6\}}$. It is worth pointing out that we obtain the same decomposition if we choose $A^\bot$ and $B^\bot$ instead, since $X_3^c = X_6$ and $Y_3^c = Y_6$.
\qed
\end{example}

It is convenient to point out that, in the considered environment, whenever we find a family of separable subcontexts whose attributes and objects  form a partition of $A$ and $B$, respectively, we have a decomposition into independent subcontexts since all the conjunctors considered in the framework do not have zero divisors. 

\begin{corollary}\label{cor:decomp}
     If  $\cF\cC\neq \varnothing$, then each pair of  $\cF\cC$  determines an independent subcontext.
\end{corollary}
\begin{proof}
    The proof straightforwardly holds from Proposition~\ref{prop:FC_ind_subcontext}.\qed
\end{proof}

Now, we will show that the other implication is true, that is, if it is possible to find a decomposition of a context into independent subcontexts, then the set $\cF\cC$ is not empty.

\begin{proposition}\label{prop:indepsubc_FN}
    If the context $(A,B,R,\sigma)$ has a decomposition into independent subcontexts $\{(A_\lambda, B_\lambda,  R_{A_\lambda\times B_\lambda}, \sigma_{A_\lambda\times B_\lambda})\}_{\lambda\in \Lambda}$, then $(\bigchi_{B_\lambda}, \bigchi_{A_\lambda})\in \cF_N$, for all $\lambda\in \Lambda$.
\end{proposition}

\begin{proof}
    Let us consider a decomposition into independent subcontexts $\{(A_\lambda, B_\lambda,  R_{A_\lambda\times B_\lambda}, \sigma_{A_\lambda\times B_\lambda})\}_{\lambda\in \Lambda}$. We are going to demonstrate the equality $\bigchi_{B_\lambda}^\uN = \bigchi_{A_\lambda}$ holds for all $\lambda\in\Lambda$. Given  $\lambda\in\Lambda$, we have that
    \begin{align*}
        \bigchi_{B_\lambda}^\uN(a) = & \inf\{\bigchi_{B_\lambda}(b)\dimp^{\sigma_\oo} R(a,b)\mid b\in B \}\\
        = & \inf\{\bot_2\dimp^{\sigma_\oo} R(a,b)\mid b\in B\setminus B_{\lambda}\}\wedge \inf\{\top_2\dimp^{\sigma_\oo} R(a,b)\mid b\in B_\lambda \}
    \end{align*}
    for all $a\in A$. By Proposition~\ref{prop:atproperties}, we obtain that $\top_2\dimp^{\sigma_\oo} R(a,b) = \top_2$, for all $b\in B_\lambda$. Thus, we have that 
    \begin{align*}
        \bigchi_{B_\lambda}^\uN(a) = & \inf\{\bot_2\dimp^{\sigma_\oo} R(a,b)\mid b\in B\setminus B_{\lambda}\}\\
    \end{align*}
    for all $a\in A$.
    Now, two cases can be distinguished depending on the subset of attributes to which $a$ belongs.
    \begin{itemize}
        \item If $a\in A_\lambda$ and $b\in B\setminus B_\lambda$, then $R(a,b) =\bot$ since $(A_\lambda, B_\lambda,  R_{A_\lambda\times B_\lambda}, \sigma_{A_\lambda\times B_\lambda})$ is a separable subcontext. Hence, by Proposition~\ref{prop:atproperties}, we obtain that
        $$\bigchi_{B_\lambda}^\uN(a) = \inf\{\bot_2\dimp^{\sigma_\oo} R(a,b)\mid b\in B\setminus B_{\lambda}\}=\top_2$$
        \item If $a\not\in A_\lambda$, then there exists $\mu\in\Lambda$ such that $a\in A_\mu$ and since the context is normalized and $(A_\mu, B_\mu,  R_{A_\mu\times B_\mu}, \sigma_{A_\mu\times B_\mu})$ is a separable subcontext, we can assert that there exists $b'\in B_\mu$ such that $R(a,b')\neq \bot$. Therefore, $\bot_2\dimp^{\sigma_\oo} R(a,b') = \bot_1$ due to the conjunctors $\adjoint^\oo_k$ have no zero-divisors.
    \end{itemize}
    Thus, we can assert that $\bigchi_{B_\lambda}^\uN = \bigchi_{A_\lambda}$, for all $\lambda\in\Lambda$. The proof of $\bigchi_{A_\lambda}^\dN = \bigchi_{B_\lambda}$ follows analogously and therefore, we can conclude that $(\bigchi_{B_\lambda}, \bigchi_{A_\lambda})\in\cF_N$, for all $\lambda\in\Lambda$.\qed
\end{proof}

Consequently, it has been established a connection between the fuzzy relation of a context and its associated Boolean relation, which lets us to know the existence of decompositions into independent subcontexts of a formal context in the fuzzy framework.

 \begin{theorem}\label{th:diis_ma_cl}
 The context $(A, B, R, \sigma)$ can be decomposed into independent subcontexts if and only if its associated Boolean context $(A, B, R^\B)$ can be decomposed into independent subcontexts.
\end{theorem}
\begin{proof}
Let us assume that the context $(A,B,R,\sigma)$ has a decomposition into independent subcontexts $\{(A_\lambda, B_\lambda,  R_{A_\lambda\times B_\lambda}, \sigma_{A_\lambda\times B_\lambda})\}_{\lambda\in \Lambda}$, where $\Lambda$ is a non-empty index set. By Proposition~\ref{prop:indepsubc_FN} and Theorem~\ref{th:fncn}, we obtain that $(B_\lambda, A_\lambda)\in\cC_N$, for all $\lambda\in\Lambda$. Moreover, since the subsets $B_\lambda$ and $A_\lambda$ form a partition of their respective sets, we can assert that $\{(A_\lambda, B_\lambda,  R^\B_{A_\lambda\times B_\lambda})\}_{\lambda\in \Lambda}$ is a decomposition into independent subcontexts of the context $(A,B,R^\B)$.

Now, let us suppose that the context $(A,B,R^\B)$ has a decomposition into independent subcontexts $\{(Y_i,X_i, R_{Y_i\times X_i})\}_{i\in I}$, where $I$ is a non-empty index set. Then, by Remark~\ref{re:dp}, we have that $(X_i,Y_i)\in\cC_N$, for all $i\in I$.
Thus, by Theorem~\ref{th:fncn}, we have that $(\bigchi_{X_i},\bigchi_{Y_i})\in \cF_N$, for all $i\in I$. Moreover, the family of pairs $\{(\bigchi_{X_i},\bigchi_{Y_i})\}_{i\in I}\subseteq \cF\cC$ and satisfy the conditions of Proposition~\ref{prop:partition}. Therefore, by Proposition~\ref{prop:FC_ind_subcontext}, we have that $\{(Y_i,X_i, R_{Y_i\times X_i},\sigma_{Y_i\times X_i})\}_{i\in I}$ is a decomposition into independent subcontexts of $(A,B,R,\sigma)$.\qed
\end{proof}

 In the following example, we present a further illustration of the preceding results.
 
\begin{example}
    Coming back to Example~\ref{ex:FN_CN}, we can provide a decomposition into independent subcontexts of the context $(A,B,R,\sigma)$ according to Definition~\ref{def:decomp_indep}, that is, $\{(A_\lambda, B_\lambda,  R_{A_\lambda\times B_\lambda}, \sigma_{A_\lambda\times B_\lambda})\}_{\lambda\in \Lambda}$, where $\Lambda =  \{1,2,3\}$, the subsets of attributes are $A_1 =\{a_1,a_2\}$, $A_2 =\{a_3\}$ and $A_3=\{a_4\}$, and the subsets of objects are $B_1=\{b_1,b_3\}$, $B_2=\{b_4\}$ and $B_3 = \{b_2\}$. The relations and the mappings are just the restriction of the fuzzy relation $R$ and the mapping $\sigma$.
    Therefore, by Proposition~\ref{prop:indepsubc_FN}, we have that the pairs $(\bigchi_{B_1}, \bigchi_{A_1}),(\bigchi_{B_2}, \bigchi_{A_2}),(\bigchi_{B_3}, \bigchi_{A_3})\in\cF_N$. Indeed, these pairs correspond to the pairs $(g_7,f_7),(g_2,f_2)$ and $(g_5,f_5)$, respectively, in Example~\ref{ex:FN_CN} and, namely the family $F_1$ in Example~\ref{ex:family_fc}. By Theorem~\ref{th:diis_ma_cl}, the associated Boolean context $(A,B,R^\B)$ has a decomposition into independent subcontexts. In fact, by Theorem~\ref{th:fncn}, we have that pairs $\{(B_\lambda, A_\lambda)\}_{\lambda\in\Lambda}\subset \cC_N$, and so a decomposition is $\{(A_\lambda, B_\lambda,  R^\B_{A_\lambda\times B_\lambda})\}_{\lambda\in \Lambda}$.

    On the other hand, we can provide a decomposition of the associated Boolean context $(A,B,R^\B)$, that is,
    $\{(A_\mu, B_\mu,  R^\B_{A_\mu\times B_\mu})\}_{\mu\in\Gamma}$, where $\Gamma =  \{1,2\}$, the subsets of attributes are $A_1 =\{a_1,a_2\}$ and $A_2=\{a_3,a_4\}$, and the subsets of objects are $B_1=\{b_1,b_3\}$ and $B_2 = \{b_2,b_4\}$. By Theorem~\ref{th:diis_ma_cl}, the context $(A,B,R,\sigma)$ has a decomposition into independent subcontexts, more precisely the decomposition is  $\{(A_\mu, B_\mu,  R_{A_\mu\times B_\mu}, \sigma_{A_\mu\times B_\mu})\}_{\mu\in\Gamma}$ which correspond to the family $F_3$ in Example~\ref{ex:family_fc}.
    \qed
\end{example}

\section{Properties of the independent subcontexts in the multi-adjoint framework}\label{sec:4}

This section presents several properties that independent subcontexts satisfy, with a particular focus on the pairs in the set $\cF\cC$. {In order to contextualize this discussion, the conditions established in the previous section shall be considered here as well.}
These properties were previously investigated in the classical setting in~\cite{COAMfac2024}, and we will show that some of these are also valid in the multi-adjoint framework. It is important to note that, as outlined in Corollary~\ref{cor:decomp}, each pair of $\cF\cC$ characterizes an independent subcontext. Therefore, all the following results are indeed properties holds by independent subcontexts.

The following result shows under what conditions a pair $(g,f)\in \mathcal{F}_N$ is related to concepts of the  multi-adjoint concept lattice associated with the context, which is denoted by $\cM$. 

\begin{proposition}\label{isc:f1}

Given a pair $(g,f)\in\cF\cC$, if $g^\up \neq f_\bot$, then
the pair $\langle g, g^\up\rangle$ is  a multi-adjoint concept, that is, $\langle g, g^\up\rangle\in \mathcal{M}$. Dually, if $f^\down\neq g_\bot$, then $\langle f^\down, f \rangle\in \mathcal{M}$.
\end{proposition}

\begin{proof}
We consider a pair $(g,f)\in\cF\cC$ with $g^\up(a_0)\neq \bot_1$ for some attribute $a_0\in A$. We need to prove that $g$ is an extent of a multi-adjoint concept, i.e., $g=g^{\up\down}$. Since the pair of operators $({}^\uparrow,{}^\downarrow)$ forms an antitone Galois connection, it is sufficient to prove that $g^{\up\down}\preceq_2 g$. Now, since $(g,f)\in\cF\cC$, we can decompose the set of objects as $B=B^\bot\cup B^\top$, where $b\in B^\top$ if $g(b)=\top_2$  and  $b\in B^\bot$ if $g(b)=\bot_2$. It is clear that the inequality $g^{\up\down}(b)\preceq_2 g(b)$ holds, for all $b\in B^\top$. 
We will show that the inequality also holds for all $b'\in B^\bot$. 
We can also decompose the set of attributes as $A=A^\bot\cup A^\top$, in the same way that we did with the set of objects with respect to $f$. Now, let us prove that the attribute $a_0$ satisfying that $g^\up(a_0)\neq\bot_1$ belongs to the set $A^\top$. 
We proceed by reductio ad absurdum, we suppose that $g^\up(a_0)\neq\bot_1$ with $a_0\in A^\bot$, then $f(a_0)=\bot_1$. In addition, we can ensure the existence of an object $b_0\in B$ such that $g(b_0)=\top_2$  and, by Lemma~\ref{lem:botrelation}, we obtain that $R(a_0, b_0)=\bot$. Therefore, since $\top_2$ is not a zero-divisor, by Proposition~\ref{prop:atproperties}~(6) and Corollary~\ref{cor:topno0}, we have that $R(a_0,b_0)\dimp^\sigma g(b_0)=\bot\dimp^\sigma \top_2=\bot_1$ and we can claim that
$$
g^\up (a_0) = \inf\{ R(a_0,b)\dimp^{\sigma} g(b)\mid b\in B\} = \bot_1
$$
which is a contradiction. Hence, we state that $g^\up(a_0)\neq\bot_1$ with $a_0\in A^\top$.
Thus, considering any $b'\in B^\bot$, since $b'\in B^\bot$ and $a_0\in A^\top$ applying Lemma~\ref{lem:botrelation}, we obtain that $R(a_0,b')=\bot$. Therefore,
we can assert that $R(a_0,b') \uimp_\sigma g^{\up}(a_0) = \bot \uimp_\sigma g^\up(a_0) = \bot_2$, since $g^\up(a_0)\neq\bot_1$ and the conjunctors have no zero-divisors. Thus, we obtain the following equality
$$
g^{\up\down}({b'}) = \inf\{R(a,b')\uimp_\sigma g^{\up}({a})\mid a\in A\} = \bot_2
$$
Therefore, we can state that $g^{\up\down}(b')\preceq_2 g(b')$ holds, for all $b'\in B^\bot$. 

As a consequence, we have  that $g^{\up\down}\preceq_2 g$, thereby ensuring that $\langle g, g^\up\rangle \in \cM$. The proof of $\langle f^\down, f \rangle\in \mathcal{M}$ is obtained analogously. \qed
\end{proof}

In accordance with the conditions previously outlined, both within the multi-adjoint framework and in the context, it is possible to ascertain the multi-adjoint concepts that constitute the top and bottom elements of the associated multi-adjoint concept lattice, $\cM$.

\begin{lemma}\label{lem:topbotconcepts}
 Given the multi-adjoint frame $(L_1, L_2, P,\adjoint_1,\dots, \adjoint_n)$ and the normalized context $(A,B,R,\sigma)$, it is satisfied that $\fcon{g_\top}{f_\bot}, \fcon{g_\bot}{f_\top}\in \cM$.
\end{lemma}
\begin{proof}
Since the concept forming operators form a Galois connection, we have that $f_\bot^{\down}=g_\top$. Now, considering any attribute $a\in A$ and applying the concept forming operator to $g_\top$, we have that
$$
g_\top^{\up}(a)= \inf\{R(a,b)\dimp^{\sigma(a,b)}  g_\top(b)\mid b\in B\}= \inf\{R(a,b)\dimp^{\sigma(a,b)}  \top_2\mid b\in B\}
$$
Since the context is normalized, for every $a\in A$ there exists $b_a\in B$ such that $R(a,b_a)=\bot$. Thus, 
$$
R(a,b_a)\dimp^{\sigma(a,b_a)}  \top_2= \bot\dimp^{\sigma(a,b_a)}  \top_2 = \max\{x\in L_1\mid x\adjoint_{\sigma(a,b_a)}  \top_2\leq \bot\}
$$
Due to the fact that $\top_2$ is not a zero-divisor, by Corollary~\ref{cor:topno0}, we can assert that the maximum of the above expression is $\bot_1$. Therefore, we obtain $g_\top^{\up}(a) =\bot_1$, for all $a\in A$. Thus, it is satisfied that  $g_\top^{\up} =f_\bot$ and $f_\bot^{\down}=g_\top$. Consequently, $\fcon{g_\top}{f_\bot}\in \cM$. Analogously, we can obtain that $\fcon{g_\bot}{f_\top}\in \cM$ considering the fuzzy-objects.\qed
\end{proof}

Furthermore, the multi-adjoint concept $\langle g, g^\up\rangle$ of Proposition~\ref{isc:f1} defines a lower-close neighbor of the top element of the concept lattice. In other words, it determines the top concept of the multi-adjoint concept lattice associated with an independent subcontext, as it is stated in the following result. 
 
\begin{proposition}\label{isc:f3}
 Given a pair $(g,f)\in \cF\cC$ with $g^\up \neq f_\bot$, there is no concept $\langle g_0, f_0\rangle$ such that 
 $$
 \langle g, g^\up\rangle \prec \langle g_0, f_0\rangle \prec \langle g_\top, f_\bot\rangle
 $$
 Dually, if $f^\down\neq g_\bot$, then there is no concept $\langle g_0, f_0\rangle$ such that 
 $$
\langle g_\bot, f_\top\rangle  \prec \langle g_0, f_0\rangle \prec \langle f^\down, f\rangle
 $$
\end{proposition}

\begin{proof}
	We will proceed by reductio ad absurdum. We assume a concept $\langle g_0, f_0\rangle$ exists such that $\langle g, g^\up\rangle \prec \langle g_0, f_0\rangle \prec \langle g_\top, f_\bot\rangle$. 
 {Since $f_0 \prec_1 g^\up$}, we have that $f_0(a) = \bot_1$, for all $a\in A^\bot$. In a similar way, it holds $g_0(b) = \top_2$, for all $b\in B^\top$ since $g\preceq_2 g_0$ and $g(b)=\top_2$, for all $b\in B^\top$. Therefore, there exists at least one object, $b_0\in B^\bot$, such that $g(b_0) = \bot_2 \preceq_2 g_0(b_0)$ and $g_0(b_0)\neq \bot_2$. 
By Lemma~\ref{lem:botrelation}, $R(a,b_0)=\bot$ for all $a\in A^\top$. Moreover, we know that $f_0(a) = \bot_1$, for all $a\in A^\bot$. Therefore,
	\begin{align*}
	 \bot_{2}\neq g_0(b_0) &=f_0^\down (b_0)  \\
	 &= \inf\{R(a,b_0)\uimp_\sigma f_0(a)\mid a\in A\}\\
	&=\inf\{R(a,b_0)\uimp_\sigma f_0(a)\mid a\in A^\top\}\wedge\inf\{R(a,b_0)\uimp_\sigma f_0(a)\mid a\in A^\bot\} \\
	&= \inf\{\bot\uimp_\sigma f_0(a)\mid a\in A^\top\}\wedge\inf\{R(a,b_0)\uimp_\sigma \bot_1\mid a\in A^\bot \}\\
	&=\inf\{\bot\uimp_\sigma f_0(a)\mid a\in A^\top\}\wedge \top_2\\
	&=\inf\{\bot\uimp_\sigma f_0(a)\mid a\in A^\top\}
	\end{align*}
    Consequently,
	\[g_0(b_0) = \inf\{\bot\uimp_\sigma f_0(a)\mid a\in A^\top\} \neq \bot_{2}\]
	Hence, if $f_0(a)\neq\bot_1$ for some $a\in A^\top$, then $\bot\uimp_\sigma f_0(a) = \bot_2$ and $g_0(b_0)= \bot_2$, which is a contradiction. Therefore, $f_0(a)=\bot_1$, for all $a\in A^\top$ and we have that $g_0(b_0) = \top_2$. Specifically,  $g_0(b) = \top_2$, for all $b\in B^\bot$.
	
	Finally, since $g_0(b) = \top_2$, for all $b\in B^\top$ and $g_0(b) = \top_2$, for all $b\in B^\bot$, we can conclude that $g_0 = g_\top$. In consequence, $\langle g_0, f_0\rangle = \langle g_\top, f_\bot\rangle$, which contradicts the hypothesis.
    The proof for the concept  $\langle f^\down, f \rangle$ can be obtained by duality. \qed
\end{proof}

In addition, for a pair $(g,f)\in\cF\cC$, when we have that both pairs $\langle g, g^\up \rangle$  and $\langle f^\down, f \rangle$
belong to $\cM$, we can state that every multi-adjoint concept associated with the independent subcontext determined by $(g,f)\in\cF\cC$ is between these two multi-adjoint concepts.
 
\begin{proposition}\label{isc:f4}
  Given a pair $(g,f)\in\cF\cC$ with $g^\up\neq f_\bot$ and $f^\down\neq g_\bot$, then the inequality $\langle f^\down, f \rangle\preceq\langle g, g^\up\rangle$ holds. 	
 \end{proposition}
\begin{proof}
Let us consider a pair $(g,f)\in\cF\cC$ with $g^\up\neq f_\bot$ and $f^\down\neq g_\bot$, and an arbitrary object $b'\in B$ in order to prove that $f^\down(b') \preceq_2 g(b')$. By Proposition~\ref{prop:partition} and Corollary~\ref{cor:compl}, we obtain partitions of the sets of object and attributes from the pair $(g,f)$, that is, $B = B^\top\cup B^\bot$ and $A = A^\top\cup A^\bot$. Therefore, two cases can be distinguished: if $b'\in B^\top$, then it is clear that $f^\down(b') \preceq_2 \top_2 = g(b')$. Otherwise, we can decompose $f^\down(b')$ using the partition of $A$ and by Proposition~\ref{prop:atproperties} (5), we have that
    \begin{align*}
         f^\down(b') =&  \inf\{ R(a,b')\uimp_{\sigma} f(a)\mid a\in A\}\\
                     =& \inf\{R(a,b')\uimp_{\sigma} \top_1\mid a\in A^\top\}\wedge\inf\{R(a,b')\uimp_{\sigma} \bot_1\mid a\in A^\bot\}\\
                     =&\inf\{R(a,b')\uimp_{\sigma} \top_1\mid a\in A^\top\}\wedge \top_1\\
                     =&\inf\{R(a,b')\uimp_{\sigma} \top_1\mid a\in A^\top\}
    \end{align*}
    Therefore, since $f(a)=\top_1$, for all $a\in A^\top$, and $g(b')=\bot_2$, by Lemma~\ref{lem:botrelation}, we have that $R(a,b')=\bot$, for all $a\in A^\top$. Thus, since $\top_1$ is not a zero-divisor by Corollary~\ref{cor:topno0}, it is satisfied that 
    $$
     f^\down(b') =\inf\{\bot\uimp_{\sigma} \top_1\mid a\in A^\top\} = \bot_2
    $$
    Hence, $f^\down(b') = g(b') = \bot_2$. Consequently, since $b'$ is an arbitrary object, we have proved that $f^\down \preceq_2 g$, and therefore, since by  Proposition~\ref{isc:f1} 
we have that   $\langle f^\down, f \rangle$ and $\langle g, g^\up\rangle$ are concepts, we can write that  $\langle f^\down, f \rangle\preceq\langle g, g^\up\rangle$.
    \qed
\end{proof}

The following example illustrates the different results previously shown.

\begin{example}\label{ex:is_prop}	
The aforementioned results can be illustrated in accordance with the setting that has been established in Example~\ref{ex:FN_CN}. Thus, considering the context $(A,B,R,\sigma)$ of Example~\ref{ex:FN_CN}, we can compute the concepts and its associated multi-adjoint concept lattice, which are both depicted in Figure~\ref{fig:prop_latt}.
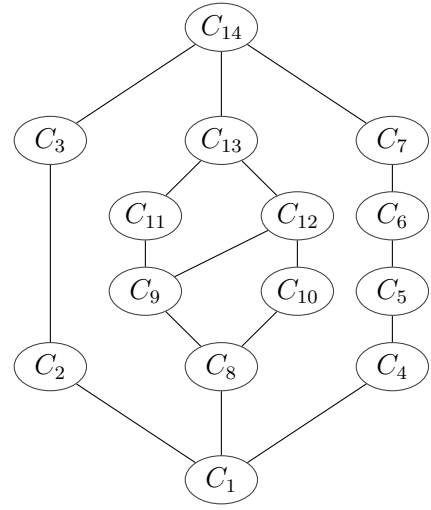
\begin{figure}[ht]
   \begin{minipage}{0.3\textwidth}
        \begin{align*}
            C_1 &= \fcon{\{\}}{\{a_1/1.0, a_2/1.0, a_3/1.0, a_4/1.0\}}\\
            C_2 &= \fcon{\{b_2/0.5\}}{\{a_4/1.0\}} \\
            C_3 &=\fcon{\{b_2/1.0\}}{\{a_4/0.4\}}\\
            C_4 &= \fcon{\{b_4/0.25\}}{\{a_3/1.0\}}\\
            C_5 &= \fcon{\{b_4/0.5\}}{\{a_3/0.6\}}\\
            C_6 &= \fcon{\{b_4/0.75\}}{\{a_3/0.4\}}\\
            C_7 &= \fcon{\{b_4/1.0\}}{\{a_3/0.2\}}\\
            C_8 &= \fcon{\{b_1/0.5, b_3/0.5\}}{\{a_1/1.0, a_2/1.0\}}\\
            C_9 &=\fcon{\{b_1/0.75, b_3/0.5\}}{\{a_1/1.0, a_2/0.8\}}\\
            C_{10} &= \fcon{\{b_1/0.5, b_3/0.75\}}{\{a_1/0.6, a_2/1.0\}}\\
            C_{11} &= \fcon{\{b_1/1.0, b_3/0.5\}}{\{a_1/1.0, a_2/0.6\}}\\
            C_{12} &= \fcon{\{b_1/0.75, b_3/1.0\}}{\{a_1/0.6, a_2/0.8\}}\\
            C_{13} &= \fcon{\{b_1/1.0, b_3/1.0\}}{\{a_1/0.6, a_2/0.6\}}\\
            C_{14} &= \fcon{\{b_1/1.0, b_2/1.0, b_3/1.0, b_4/1.0\}}{\{\}}
        \end{align*}
    \end{minipage}
    \begin{minipage}{0.5\textwidth}
        \begin{center}  
        \tikzstyle{place}=[ellipse, align=center,minimum width=30pt, draw=black!75,fill=white!20, text width= 17pt]
        \begin{tikzpicture}[inner sep=0.75mm,scale=1.0, every node/.style={scale=0.9}]			
            \node at (0,-0.5) (0) [place] {$C_1$};
            \node at (-2.25,1) (1) [place] {$C_2$};
            \node at (-2.25,4) (2) [place] {$C_3$};
            \node at (2.25,1) (3) [place] {$C_4$};
            \node at (2.25,2) (4) [place] {$C_5$};
            \node at (2.25,3) (5) [place] {$C_6$};
            \node at (2.25,4) (6) [place] {$C_7$};
            \node at (0,1) (7) [place] {$C_8$};
            \node at (-1,2) (8) [place] {$C_9$};
            \node at (1,2) (9) [place] {$C_{10}$};
            \node at (-1,3) (10) [place] {$C_{11}$};
            \node at (1,3) (11) [place] {$C_{12}$};
            \node at (0,4) (12) [place] {$C_{13}$};
            \node at (0,5.5) (13) [place] {$C_{14}$};
        
            \draw [-] (0) -- (1)--(2)--(13);
            \draw [-] (0) -- (3)--(4)--(5)--(6)--(13);
            \draw [-] (0) -- (7)--(8)--(10)--(12)--(13);
            \draw [-] (8)--(11);
            \draw [-] (7)--(9)--(11)--(12);
        \end{tikzpicture}
        \end{center}
    \end{minipage}
\caption{List of multi-adjoint concepts of the context $(A,B,R,\sigma)$ and its associated multi-adjoint concept lattice $\cM$ in Example~\ref{ex:is_prop}.}
\label{fig:prop_latt}
\end{figure}

Recall that all the elements of the set $\cF\cC$ were listed in Example~\ref{ex:FN_CN_th} and we are considering the pair $(g_7,f_7)\in\cF\cC\subset\cF_N$, where 
$$(g_7,f_7) = (\{b_1/1.0, b_3/1.0\}, \{a_1/1.0,a_2/1.0\})$$
It is easy to check that $g_7^\up = \{a_1/0.6, a_2/0.6\}$, and therefore, we are under the hypothesis of Proposition~\ref{isc:f1}, that is, $g_7^\up(a_2) \neq \bot_{1}$. Thus, by Proposition~\ref{isc:f1}, $\langle g_7, g_7^\up\rangle$ is a multi-adjoint concept. Indeed, it is the concept $C_{13}$ listed in Figure~\ref{fig:prop_latt} and, moreover, it is a lower-close neighbor of the top element of the concept lattice, as Proposition~\ref{isc:f3} states. Furthermore, $\langle f_7^\down, f_7\rangle$ is also a multi-adjoint concept, namely the concept $C_8$, and therefore, by Proposition~\ref{isc:f4}, an interval of concepts is determined from this two concepts, i.e., $[C_8, C_{13}]= [\fcon{f_7^\down}{f_7}, \fcon{g_7}{g_7^\up}] = \{\fcon{g}{f}\in\cM\mid \fcon{f_7^\down}{f_7}\} \preceq \fcon{g}{f}\preceq \fcon{g_7}{g_7^\up}\}$. A similar outcome is yielded when the pairs  $(g_2,f_2)$ and $(g_5,f_5)$ are taken into consideration.\qed
\end{example}

Therefore, when the considered context  contains independent subcontexts, the previous results allow us to know the top and bottom elements of the concept lattice associated with an independent subcontext.

\begin{corollary}\label{cor:topbot}
 Given a pair $(g,f)\in\cF\cC$ with $g^\up\neq f_\bot$ and $f^\down\neq g_\bot$,
then the concepts $\langle g, g^\up\rangle$ and $\langle f^\down, f \rangle$ determine the top and bottom element, respectively, of the concept lattice associated with an independent subcontext. 
 \end{corollary}
 \begin{proof}
 Straightforwardly from Theorem~\ref{th:diis_ma_cl} , Proposition~\ref{isc:f1}, Proposition~\ref{isc:f3} and Proposition~\ref{isc:f4}.\qed
 \end{proof}

\section{A procedure to decompose a context}\label{sec:5}

It is possible that the Boolean context associated with a context does not  contain independent subcontexts. In this case, it would be interesting to study the possibility of  finding independent subcontexts by making some modifications on the fuzzy relation of the original context. In this section, we provide a three-step procedure to know if a context, $(A,B,R,\sigma)$, whose associated Boolean context does not  contain independent subcontexts, can be modified in order to contain independent subcontexts. This procedure is based on the use of thresholds which ``remove'' (weak) relations of the fuzzy relation $R$ of the context. The steps to follow in the procedure are detailed below:

\begin{enumerate}[label=\textbf{Step \arabic*:}, leftmargin=*] 

\item Fix the largest possible value  $\alpha\in P$ such that the relations $R_\alpha\colon A\times B\to P$, defined for all $a\in A$ and $b\in B$  as:

\[R_\alpha(a,b) = \left\{\begin{array}{ll}
	R(a,b) & \mbox{ if } \alpha\leq R(a,b) \\
	\bot & \mbox{ otherwise}  
\end{array}\right.\]
preserves the context normalized, that is, satisfies that the relationship $R_\alpha$ has neither rows nor columns with all the values equal to $\bot$.
\item Build the associated Boolean context of the context $(A,B,R_\alpha,\sigma)$. 

\item Compute all possible pairs of the set  $\cF\cC$ through the context $(A, B, R_\alpha^\B)$ since, according to Theorem~\ref{th:diis_ma_cl}, when independent subcontexts in the Boolean context can be found, then the context $(A,B,R_\alpha,\sigma)$ also contains independent subcontexts. 
\end{enumerate}

Notice that, since the considered relation $R$ in the multi-adjoint framework has neither rows nor columns full with values $\bot$ {and the set $P$ has a finite domain}, the fixed value $\alpha$ will be always different from $\bot$.

It is also convenient to highlight that the way in which the value of $\alpha$ has been fixed in the step one of the previous procedure, guarantees that it is considered the maximun possible value of $\alpha$ which can provide a Boolean context with independent subcontexts.
This is due to, in such a case,  the context $(A,B,R_\alpha,\sigma)$ would contain either rows or columns with all values equal to $\bot$.
	
While it is true that this largest possible value of alpha, is the value that causes  a largest number of changes in the original context, it is important to start from this $\alpha$ because of  if the  obtained context  $(A,B,R_\alpha,\sigma)$  does not contain independent subcontexts, then it will not be possible to find independent subcontexts to any smaller  value of alpha.

Moreover, when the obtained context  $(A,B,R_\alpha,\sigma)$  has independent subcontexts, it is possible to try to reduce the impact of the considered  value~$\alpha$, taking a smaller value   and repeating the procedure for this new  value.

In the following example we show the  application of the previous procedure.	
\begin{example}\label{ex:proc_facto}
Let us consider the multi-adjoint frame $(L_1,L_2,P,\adjoint^*_\G,\adjoint^*_\Pp)$, the property-oriented multi-adjoint frame $(L_1,L_2,P,\adjoint^*_\G)$ and the object-oriented multi-adjoint frame $(L_1,L_2,P,\adjoint^*_\Pp)$, where $L_1 = L_2 = P = [0,1]_4$.
We consider the context $(A, B, R,\sigma)$ given in Table~\ref{tab:r_original}. This context cannot be decomposed into independent subcontexts according to Definition~\ref{def:decomp_indep} since the objects $b_1$ and $b_3$ are related to all attributes. In other words, the context is not normalized.
	\begin{table}[ht]
    \centering
        \begin{minipage}{0.45\textwidth}
            \centering
            \begin{tabular}{|c|c c c c|}
            \hline
				$R$ & $b_1$ &$b_2$&$b_3$&$b_4$\\ \hline
				$a_1$&1&0.5&0.25&0\\
				$a_2$&0.5&0.75&0.25&0\\
				$a_3$&0.25&0&0.75&0.5\\
				$a_4$&0.25&0&0.75&0.25\\
				$a_5$&0.25&0.25&0.5&1\\
				\hline
        \end{tabular}
        \end{minipage}
        \begin{minipage}{0.45\textwidth}
        \centering
            \begin{tabular}{|c|c c c c|}
                     \hline
                     $\sigma$ & $b_1$ & $b_2$ & $b_3$ & $b_4$  \\
                     \hline
                     $a_1$ & $\adjoint^*_\G$ & $\adjoint^*_\G$ & $\adjoint^*_\G$ & $\adjoint^*_\G$ \\
                     $a_2$ & $\adjoint^*_\Pp$ & $\adjoint^*_\Pp$ & $\adjoint^*_\Pp$ & $\adjoint^*_\Pp$ \\
                     $a_3$ & $\adjoint^*_\G$ & $\adjoint^*_\G$ & $\adjoint^*_\G$ & $\adjoint^*_\G$ \\
                     $a_4$ & $\adjoint^*_\Pp$ & $\adjoint^*_\Pp$ & $\adjoint^*_\Pp$ & $\adjoint^*_\Pp$ \\
                     $a_5$ & $\adjoint^*_\G$ & $\adjoint^*_\G$ & $\adjoint^*_\G$ & $\adjoint^*_\G$ \\
                     \hline
                \end{tabular}
        \end{minipage}		
		\caption{Fuzzy relation $R$ and the mapping $\sigma$ of the context $(A,B,R,\sigma)$ in Example~\ref{ex:proc_facto}.}\label{tab:r_original}
	\end{table}	

Therefore, we can apply the procedure to find a new formal context with independent subcontexts.
\begin{enumerate}

\item In this case, we set the value $\alpha=0.75$ which is the maximum value in $R$ that does not make either rows or columns with all values equal to $0$. Thus, we obtain the fuzzy relation $R_{0.75}$  given on the left side in Table~\ref{tab:R075}.

	\begin{table}[ht]
	\begin{minipage}{0.55\textwidth}
		\begin{center}
			\begin{tabular}{|c|cccc|}
				\hline
				$R_{0.75}$ & $b_1$ &$b_2$&$b_3$&$b_4$\\ \hline
				$a_1$&1&0&0&0\\
				$a_2$&0&0.75&0&0\\
				$a_3$&0&0&0.75&0\\
				$a_4$&0&0&0.75&0\\
				$a_5$&0&0&0&1\\
				\hline
			\end{tabular}
		\end{center}
	\end{minipage}
	\begin{minipage}{0.5\textwidth}
		\begin{center}
			\begin{tabular}{|c|cccc|}
				\hline
				$R_{0.75}^\B$ & $b_1$ &$b_2$&$b_3$&$b_4$\\ \hline
				$a_1$&1&0&0&0\\
				$a_2$&0&1&0&0\\
				$a_3$&0&0&1&0\\
				$a_4$&0&0&1&0\\
				$a_5$&0&0&0&1\\
                \hline
			\end{tabular}
		\end{center}
	\end{minipage}
	\vspace*{1em}
	\caption{Fuzzy relation $R_{0.75}$  and its associated Boolean relation $R_{0.75}^{\B}$ in Example~\ref{ex:proc_facto}}\label{tab:R075}
\end{table}

\item Once we have the fuzzy relation $R_{0.75}$, we compute its associated Boolean relation $R_{0.75}^{\B}$ given on the right side of Table~\ref{tab:R075}. 

\item The  elements of the set  $\cF\cC$ are derived from the context $(A,B,R_{0.75}^{\B})$, which gives rise to the following pairs:

\begin{minipage}{0.25\textwidth}
  \begin{align*}
	 (X_1, Y_1)&=(\{b_1\},\{a_1\})\\
	 (X_2,Y_2)&=(\{b_2\},\{a_2\})\\
	 (X_3,Y_3)&=(\{b_3\},\{a_3,a_4\})\\
	 (X_4,Y_4)&=(\{b_4\},\{a_5\})\\
     (X_5, Y_5)&=(\{b_1,b_2\},\{a_1,a_2\})\\
     (X_6, Y_6)&=(\{b_1,b_3\},\{a_1,a_3,a_4\})\\
     (X_7, Y_7)&=(\{b_1,b_4\},\{a_1,a_5\})\\
    \end{align*}  
\end{minipage}
\begin{minipage}{0.3\textwidth}
  \begin{align*}
	 (X_8, Y_8)&=(\{b_2,b_3\},\{a_2,a_3,a_4\})\\
	 (X_9,Y_9)&=(\{b_2,b_4\},\{a_2,a_5\})\\
	 (X_{10},Y_{10})&=(\{b_3,b_4\},\{a_3,a_4,a_5\})\\
	 (X_{11},Y_{11})&=(\{b_1,b_2,b_3\},\{a_1,a_2,a_3,a_4\})\\
     (X_{12}, Y_{12})&=(\{b_1,b_2,b_4\},\{a_1,a_2,a_5\})\\
     (X_{13}, Y_{13})&=(\{b_1,b_3,b_4\},\{a_1,a_3,a_4,a_5\})\\
     (X_{14}, Y_{14})&=(\{b_2,b_3,b_4\},\{a_2,a_3,a_4,a_5\})\\
\end{align*}   
\end{minipage}

	As a consequence, in this case, we can find fourteen  distinct independent subcontexts. Several of these contexts satisfy the hypotheses of Proposition~\ref{isc:f4}, as it can be  observed in the concept lattice in Figure~\ref{cl:ex_pro1}. For example, $(X_2,Y_2)$ determines the interval given by $C_3$ and $C_4$, that is,
    $$
   \langle (\bigchi_Y)^\down, \bigchi_Y \rangle=C_3\preceq C_4=\langle \bigchi_X, (\bigchi_X)^\up\rangle 
    $$  	
		
	{	\renewcommand{\arraystretch}{1.5}
	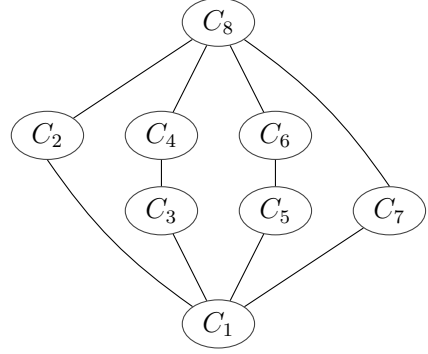
\begin{figure}[ht]
		\begin{minipage}{0.55\textwidth}
		
				\begin{tabular}{l}
						$C_1 =\langle\{\}, \{a_1/1.0 , a_2/1.0 , a_3/1.0 , a_4/1.0 , a_5/1.0\}\rangle$\\
						$C_2 = \langle\{b_1/1.0\} , \{a_1/1.0\}\rangle$\\
						$C_3 = \langle\{b_2/0.75 \} , \{a_2/1.0\}\rangle$\\
						$C_4 = \langle\{b_2/1.0\} , \{a_2/0.75\}\rangle$\\ 
						$C_5 = \langle\{b_3/0.75\} , \{a_3/1.0 , a_4/1.0\}\rangle$\\ 
						$C_6 = \langle\{b_3/1.0\} , \{a_3/0.75, a_4/0.75\}\rangle$\\
						$C_7 = \langle\{b_4/1.0\} , \{a_5/1.0\}\rangle$\\
                        $C_8 =\langle\{b_1/1.0 , b_2/1.0, b_3/1.0, b_4/1.0\} , \{ \}\rangle$
					\end{tabular}
			\end{minipage}
		\begin{minipage}{0.3\textwidth}
        \begin{center}  
        \tikzstyle{place}=[ellipse, align=center,minimum width=30pt, draw=black!75,fill=white!20, text width= 17pt]
        \begin{tikzpicture}[inner sep=0.75mm,scale=1.0, every node/.style={scale=0.9}]			
            \node at (0,-0.5) (0) [place] {$C_1$};
            \node at (-2.25,2) (1) [place] {$C_2$};
            \node at (-0.75,1) (2) [place] {$C_3$};
            \node at (-0.75,2) (3) [place] {$C_4$};
            \node at (0.75,1) (4) [place] {$C_5$};
            \node at (0.75,2) (5) [place] {$C_6$};
            \node at (2.25,1) (6) [place] {$C_7$};
            \node at (0,3.5) (7) [place] {$C_8$};
        
            \draw [-] (1)--(7)--(3)--(2)--(0)--(4)--(5)--(7);
            \draw [-] (0)--(6);
            \draw [-] (0.north west) to[bend left=10] (1.south);
            \draw [-] (6.north) to[bend right=10] (7.south east);
        \end{tikzpicture}
        \end{center}
    \end{minipage}
	\caption{Concepts and the multi-adjoint concept lattice of the context $(A,B,R_{0.75},\sigma)$.}\label{cl:ex_pro1}
	\end{figure}}
\end{enumerate} 

As we can observe from Table~\ref{tab:R075} and Figure~\ref{cl:ex_pro1}, the consideration of the value $\alpha=0.75$ has confirms the possibility of obtaining a decomposition of the original context, but also has entailed a considerable number of modifications of the original context and an important reduction of the size of the original concept lattice (which  has 33 concepts). Hence, we can try to reduce this impact, considering a new and smaller  value of $\alpha$. In this case, we can take $\alpha = 0.5$ and repeat the procedure. For the context $(A,B,R_{0.5},\sigma)$ whose relation is given on the left side of Table~\ref{ex:alpha2}, we consider its associated Boolean context $(A, B, R_{0.5}^\B)$, where its Boolean relation is given on the right side of Table~\ref{ex:alpha2}.

	\begin{table}[ht]
	\begin{minipage}{0.55\textwidth}
		\begin{center}
			\begin{tabular}{|c|cccc|}
				\hline
				$R_{0.5}$ & $b_1$ &$b_2$&$b_3$&$b_4$\\ \hline
				$a_1$&1&0.5&0&0\\
				$a_2$&0.5&0.75&0&0\\
				$a_3$&0&0&0.75&0.5\\
				$a_4$&0&0&0.75&0\\
				$a_5$&0&0&0.5&1\\
				\hline
			\end{tabular}
		\end{center}
	\end{minipage}
	\begin{minipage}{0.5\textwidth}
		\begin{center}
			\begin{tabular}{|c|cccc|}
				\hline
				$R_{0.5}^\B$ & $b_1$ &$b_2$&$b_3$&$b_4$\\ \hline
				$a_1$&1&1&0&0\\
				$a_2$&1&1&0&0\\
				$a_3$&0&0&1&1\\
				$a_4$&0&0&1&0\\
				$a_5$&0&0&1&1\\
				\hline
			\end{tabular}
		\end{center}
	\end{minipage}
	\vspace*{1em}
	\caption{Fuzzy relation $R_{0.5}$ (left) and its associated Boolean relation $R_{0.5}^{\B}$ (right).}\label{ex:alpha2}
\end{table}	
For this smaller value of $\alpha$, we obtain less elements of the set $\cF\cC$ from this Boolean context, that means we find less independent subcontexts, as it can be observed  in Figure~\ref{cl:ex_pro2}. In this case, the elements of $\cF\cC$ are obtained from the following pairs:

\begin{itemize}
	\item $(X'_1,Y'_1)=(\{b_1,b_2\},\{a_1,a_2\})$.
	\item $(X'_2,Y'_2)=(\{b_3,b_4\},\{a_3,a_4,a_5\})$.
\end{itemize}
	{	\renewcommand{\arraystretch}{1.5}
	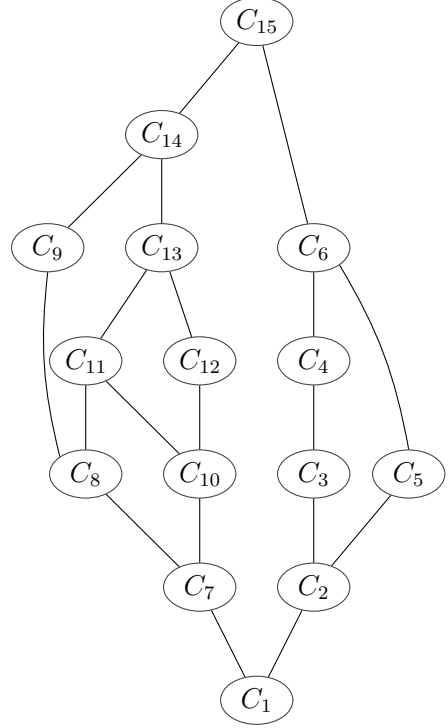
\begin{figure}[ht]
		\begin{minipage}{0.6\textwidth}
			\hspace*{-0.5em}
			\begin{tabular}{l}
				$C_1 =\langle\{\}, \{a_1/1.0 , a_2/1.0 , a_3/1.0 , a_4/1.0 , a_5/1.0\}\rangle$\\
				$C_2 = \langle\{b_1/0.5 , b_2/0.5 \} , \{a_1/1.0 , a_2/1.0\}\rangle$\\
				$C_3 = \langle\{b_1/0.5 , b_2/0.75 \} , \{a_1/0.5 , a_2/1.0\}\rangle$\\
                $C_4 = \langle\{b_1/0.5 , b_2/1.0 \} , \{a_1/0.5 , a_2/0.75\}\rangle$\\
                $C_5 = \langle\{b_1/1.0 , b_2/0.5 \} , \{a_1/1.0 , a_2/0.5\}\rangle$\\
                $C_6 = \langle\{b_1/1.0 , b_2/1.0 \} , \{a_1/0.5 , a_2/0.5\}\rangle$\\
				$C_7 = \langle\{b_3/0.5 \} , \{a_3/1.0 , a_4/1.0 , a_5/1.0\}\rangle$\\
				$C_8 = \langle\{b_3/0.5, b_4/0.5 \} , \{a_3/1.0 , a_5/1.0\}\rangle$\\
                $C_9 = \langle\{b_3/0.5, b_4/1.0 \} , \{a_3/0.5 , a_5/1.0\}\rangle$\\
                $C_{10} = \langle\{b_3/0.75 \} , \{a_3/1.0, a_4/1.0, a_5/0.5\}\rangle$\\
                $C_{11} = \langle\{b_3/0.75, b_4/0.5 \} , \{a_3/1.0 , a_5/0.5\}\rangle$\\
                $C_{12} = \langle\{b_3/1.0 \} , \{a_3/0.75, a_4/0.75, a_5/0.5\}\rangle$\\
                $C_{13} = \langle\{b_3/1.0, b_4/0.5 \} , \{a_3/0.75 , a_5/0.5\}\rangle$\\        
				$C_{14} =\langle\{b_3/1.0 , b_4/1.0\} , \{a_3/0.5 , a_5/0.5\}\rangle$\\
                $C_{15} = \langle\{b_1/1.0 , b_2/1.0, b_3/1.0, b_4/1.0\} , \{ \}\rangle$\\
			\end{tabular}
		\end{minipage}
		\begin{minipage}{0.25\textwidth}
		\begin{center}  
        \tikzstyle{place}=[ellipse, align=center,minimum width=30pt, draw=black!75,fill=white!20, text width= 17pt]
        \begin{tikzpicture}[inner sep=0.75mm,scale=1.0, every node/.style={scale=0.9}]			
            \node at (0,-0.5) (1) [place] {$C_1$};
            \node at (0.75,1) (2) [place] {$C_2$};
            \node at (0.75,2.5) (3) [place] {$C_3$};
            \node at (0.75,4) (4) [place] {$C_4$};
            \node at (2,2.5) (5) [place] {$C_5$};
            \node at (0.75,5.5) (6) [place] {$C_6$};
            \node at (-0.75,1) (7) [place] {$C_7$};
            \node at (-2.25,2.5) (8) [place] {$C_8$};
            \node at (-0.75,2.5) (9) [place] {$C_{10}$};
            \node at (-2.25,4) (10) [place] {$C_{11}$};
            \node at (-0.75,4) (11) [place] {$C_{12}$};
            \node at (-1.25,5.5) (12) [place] {$C_{13}$};
            \node at (-2.75,5.5) (13) [place] {$C_{9}$};
            \node at (-1.25,7) (14) [place] {$C_{14}$};
             \node at (0,8.5) (15) [place] {$C_{15}$};
        
            \draw [-] (1)--(2)--(3)--(4)--(6)--(15);
            \draw [-] (5.north) to[bend right=10] (6.south east);
            \draw [-] (8.north west) to[bend left=10] (13.south);
            \draw [-] (1) -- (7)--(9)--(11)--(12)--(14)--(15);
            \draw [-] (7) -- (8)--(10)--(12);
            \draw [-] (2)--(5);
            \draw [-] (9)--(10);
            \draw [-] (13)--(14);
        \end{tikzpicture}
        \end{center}
		\end{minipage}
		\caption{Concepts and the multi-adjoint concept lattice of the context $(A,B,R_{0.5},\sigma)$.}\label{cl:ex_pro2}
\end{figure}}

Notice that, from the concepts listed on the left side of Figure~\ref{cl:ex_pro2}, we can verify that intervals of concepts in the concept lattice can be determined from the pairs of $\cF\cC$ as Corollary~\ref{cor:topbot} states. One of them is delimited by concepts $C_2$ and $C_6$ obtained by $(\bigchi_{X'_1}, \bigchi_{Y'_1})$ and the other one by $C_7$ and $C_{14}$ obtained by $(\bigchi_{X'_2}, \bigchi_{Y'_2})$.

Therefore, with the threshold $\alpha=0.5$, we have obtained a new decomposable context that preserves a greater amount of information from the original context. \qed
\end{example}

\section{Conclusions and future work}\label{sec:conclu}

This paper has presented a mechanism for detecting and computing  
 independent subcontexts appearing in a context, extending the procedure given in the classical case based on the  necessity operators given in possibility theory~\cite{dubois:2012}. 
 
Moreover, several properties of these operators have been proved. For example, we have proved that each pair in $ \cF\cC$ provides the top and bottom concepts of the corresponding independent subcontext and that no other concept exists between them and the top and bottom concepts of the whole concept lattice of the original context.  In addition, we have proposed a procedure in order to decomposed a context when it is not possible to obtain independent subcontexts.
{This last mechanism can be seen as a procedure to compute ``approximate''  independent subcontexts or to avoid the consideration of low relations, which could be understood as non-representative, providing more noise than information}

In the future, more properties will be studied and new mechanisms to determine independent or pseudo independent parts in a dataset will be analyzed. 
{Furthermore, the possibility of improving the factorization process~\cite{BelohlavekTFS24} using the  decomposition procedure introduced in this paper will also be studied.}
These results will be applied to real datasets, such as the ones given by Grupo Energ\'etico de Puerto Real (Spain) on renewable energy (photovoltaic facilities)~\cite{Gen2022_CL,Gen2024Mod1} or the ones related to digital forensic within the COST Action DigForASP~\cite{SCI24Choquet}.

\section*{Declarations}

The authors have no competing interests to declare that are relevant to the content of this article.

\end{document}